\pdfoutput=1
\documentclass[%
 reprint,
 amsmath,amssymb,
 aps,
 nofootinbib,
 superscriptaddress,
 longbibliography,
 showkeys
]{revtex4-2}

\usepackage{amsthm}
\usepackage{graphicx}
\usepackage{xcolor}
\usepackage{physics}
\usepackage{braket}
\usepackage{amsmath}
\usepackage{amssymb}
\usepackage{enumerate}
\usepackage{mathtools}
\usepackage{etoolbox}
\usepackage{algorithmic}
\usepackage{algorithm}
\usepackage{bbm}
\usepackage{here}
\usepackage{float}
\usepackage{comment}
\usepackage{bm}
\usepackage{acronym}
\usepackage[colorlinks=true
,urlcolor=DARKBLUE
,anchorcolor=DARKBLUE
,citecolor=DARKBLUE
,filecolor=DARKBLUE
,linkcolor=DARKBLUE
,menucolor=DARKBLUE
,linktocpage=true
,pdfproducer=medialab
,pdfa=true
]{hyperref}
\usepackage{subfigure}
\usepackage{siunitx}

\newtheorem{theorem}{Theorem}
\newtheorem{lemma}{Lemma}
\newtheorem{problem}{Problem}
\newtheorem{definition}{Definition}
\newtheorem{assumption}{Assumption}
\newtheorem{corollary}{Corollary}

\newcommand{\proc}[1]{\textup{\textsf{#1}}}

\newcommand{\beae}[1]{\begin{equation}\begin{aligned} #1 \end{aligned}\end{equation}}
\newcommand{\bege}[1]{\begin{equation}\begin{gathered} #1 \end{gathered}\end{equation}}
\newcommand{\bae}[1]{\begin{align} #1 \end{align}}

\newcommand{\bme}[1]{\begin{multline} #1 \end{multline}}

\newcommand{\bmbe}[1]{\begin{multlined}[b] #1 \end{multlined}}

\definecolor{MONZA}{HTML}{CF000F}
\definecolor{DARKBLUE}{HTML}{00008b}
\definecolor{DARKMAGENTA}{HTML}{8b008b}

\newcommand{\ee}{\mathrm{e}}
\newcommand{\Mpl}{M_\mathrm{Pl}}
\newcommand{\FP}{\mathrm{FP}}
\newcommand{\FPT}{\mathrm{FPT}}
\newcommand{\umax}{\mathrm{max}}

\newcommand{\sto}{\mathrm{sto}}

\newcommand{\uc}{\mathrm{c}}
\newcommand{\uf}{\mathrm{f}}
\newcommand{\calL}{\mathcal{L}}
\newcommand{\calM}{\mathcal{M}}
\newcommand{\calN}{\mathcal{N}}
\newcommand{\calP}{\mathcal{P}}
\newcommand{\bbR}{\mathbb{R}}
\newcommand{\bfx}{\mathbf{x}}

\begin{document}

\title{Improved quantum algorithm for calculating eigenvalues of differential operators and its application to estimating the decay rate of the perturbation distribution tail in stochastic inflation}
\author{Koichi Miyamoto}
\email{miyamoto.kouichi.qiqb@osaka-u.ac.jp}
\affiliation{Center for Quantum Information and Quantum Biology, The University of Osaka, Toyonaka 560-0043, Japan}

\author{Yuichiro Tada}
\email{tada.yuichiro.y8@f.mail.nagoya-u.ac.jp}
\affiliation{Institute for Advanced Research, Nagoya University,
Furo-cho, Chikusa-ku, 
Nagoya 464-8601, Japan}
\affiliation{Department of Physics, Nagoya University,
Furo-cho, Chikusa-ku,
Nagoya 464-8602, Japan}

\date{\today}

\begin{abstract}

Quantum algorithms for scientific computing and their applications have been studied actively.
In this paper, we propose a quantum algorithm for estimating the first eigenvalue of a differential operator $\mathcal{L}$ on $\mathbb{R}^d$ and its application to cosmic inflation theory.
A common approach for this eigenvalue problem involves applying the finite-difference discretization to $\mathcal{L}$ and computing the eigenvalues of the resulting matrix, but this method suffers from the curse of dimensionality, namely the exponential complexity with respect to $d$.
Our first contribution is the development of a new quantum algorithm for this task, leveraging recent quantum singular value transformation-based methods.
Given a trial function that overlaps well with the eigenfunction, our method runs with query complexity scaling as $\widetilde{O}(d^3/\epsilon^2)$ with $d$ and estimation accuracy $\epsilon$, which is polynomial in $d$ and shows an improvement over existing quantum algorithms.
Then, we consider the application of our method to a problem in a theoretical framework for cosmic inflation known as stochastic inflation, specifically calculating the eigenvalue of the adjoint Fokker--Planck operator, which is related to the decay rate of the tail of the probability distribution for the primordial density perturbation.
We numerically see that in some cases, simple trial functions overlap well with the first eigenfunction, indicating our method is promising for this problem.

\end{abstract}

\maketitle

\acrodef{PDF}{probability density function}
\acrodef{FP}{Fokker--Planck}
\acrodef{PDE}{partial differential equation}
\acrodef{QSVT}{quantum singular value transformation}

\section{Introduction}

Today, we are witnessing the rapid development of quantum computing.
Although it may take decades for us to obtain large-scale fault-tolerant quantum computers, to prepare for such a future, the search for applications of quantum algorithms to practical numerical problems across various fields has increasingly gained momentum in recent years.
In this paper, we consider a quantum algorithm for calculating eigenvalues of differential operators and its application to a problem in cosmology.

With the Sturm--Liouville problem as a representative example, the calculation of eigenvalues of linear partial differential operators, especially the first (that is, smallest) eigenvalue, is one of the major topics in the field of numerical analysis for \acp{PDE} (see Ref.~\cite{larsson2003partial} as a textbook).
As is the case for many problems in numerical analysis, this problem suffers from the so-called \emph{curse of dimensionality}.
To calculate the eigenvalues of a linear partial differential operator $\mathcal{L}$ acting on functions on $\mathbb{R}^d$, a straightforward approach is the finite difference approximation: setting grid points in $\mathbb{R}^d$, we approximate partial derivatives with difference quotients.
This converts $\mathcal{L}$ to a matrix $L$ with finite size, and then we can apply some numerical method for matrix eigenvalue calculation to $L$ (see, e.g., Ref.~\cite{Demmel1997}).
However, the size of $L$ becomes exponentially large with respect to the dimension $d$, that is, $n_\mathrm{gr}^d \times n_\mathrm{gr}^d$, where $n_\mathrm{gr}$ is the number of the grids in one dimension.
Since classical algorithms to calculate eigenvalues of a $n \times n$ matrix generally have the complexity scaling as $O(\mathrm{poly}(n))$, the above approach does not work for large $d$ in classical computing.  

% Classical eigenvalue calculation algorithm: QR iteration for an $n \times n$ tridiagonal matrix: $O(n)$ \cite{Demmel1997}. %p.213

Quantum computers may provide a solution to this issue.
Represented by the monumental Harrow--Hassidim--Lloyd algorithm for solving linear equation systems~\cite{HHL}, quantum algorithms for various linear algebra problems have been devised.
Algorithms for eigenvalue calculation are among them~\cite{Abrams1999,kerenidis2017,chakraborty2019,Gilyen2019,Ge2019,Lin2020nearoptimalground,Martyn2021,low2024quantum}.
They achieve an exponential speed-up compared to the classical ones: their complexities scale as $O(\mathrm{polylog}(n))$.
Therefore, we are motivated to apply these quantum algorithms to the aforementioned finite difference approach for differential operator eigenvalue problems.
In fact, such quantum methods have already been considered in the previous papers~\cite{Szkopek2005,papageorgiou2005classical,BESSEN2006660}.
However, to the best of our knowledge, quantum algorithms for differential operator eigenvalue problems have been studied only in the above papers in 2000s, and among them, only Ref.~\cite{Szkopek2005} considered the multi-dimensional setting and evaluated the complexity of their algorithm in it.
Thus, recent developments in quantum algorithms have not been reflected in this problem yet.

In this paper, we first construct a quantum algorithm for calculating differential operator eigenvalues in a modern way.
Concretely, we utilize the quantum algorithm for calculating matrix eigenvalues devised in Ref.~\cite{Lin2020nearoptimalground}.
This algorithm is based on the recent techniques of block-encoding and \ac{QSVT}~\cite{Low2017,Low2019hamiltonian,chakraborty2019,Gilyen2019,Martyn2021}, which have been now foundations of various quantum algorithms.
We then combine this algorithm with the finite difference approach for differential operator eigenvalues.
Given a test function that overlaps the first eigenfunction well, a quantum circuit to generate the quantum state encoding that function and circuits to calculate the coefficients in $\mathcal{L}$, our algorithm outputs an estimate of the first eigenvalue, where the above circuits are queried a number of times scaling polynomially with respect to $d$ and as $\widetilde{O}(1/\epsilon^2)$\footnote{$\widetilde{O}(\cdot)$ hides the logarithmic factors in the Landau's big-O notation.} with respect to the estimation accuracy $\epsilon$.
Our quantum algorithm achieves an exponential speed-up with respect to $d$ compared to classical algorithms and improvement with respect to $\epsilon$ compared to the method in Ref.~\cite{Szkopek2005}, whose query complexity is $\widetilde{O}(1/\epsilon^3)$.

As the second contribution, we consider a problem in cosmology, the estimation of the decay rate of the tail of the probability distribution of the perturbation in stochastic inflation.
Cosmic inflation~\cite{Starobinsky:1980te,Sato:1981qmu,Guth:1980zm,Linde:1981mu,Albrecht:1982wi,Linde:1983gd}, the accelerated expansion of the early universe, was introduced as a solution to some problems in the Big Bang cosmology and has been a standard paradigm today.
It is considered that the quantum fluctuation of inflatons, the scalar fields that induced inflation, leads to the primordial density perturbation, which is the seed of today's rich structures of the universe such as galaxies.
Stochastic inflation~\cite{Starobinsky:1982ee,Starobinsky:1986fx,Nambu:1987ef,Nambu:1988je,Kandrup:1988sc,Nakao:1988yi,Nambu:1989uf,Mollerach:1990zf,Linde:1993xx,Starobinsky:1994bd} is the formalism to analyze the inflationary perturbation in a probabilistic way.
In stochastic inflation, the fluctuations of the inflatons are associated with a \ac{PDE} called the adjoint \ac{FP} \ac{PDE}, and the eigenvalues of the differential operator in it determine the decay rate of the tail of the perturbation distribution.
In particular, as an interesting scenario, if there are small eigenvalues, the tail becomes fat, and primordial black holes, a candidate for dark matter, are produced abundantly, which motivates us to estimate the first eigenvalue.
Our quantum algorithm may be thus useful for this task, especially in multi-field inflation, where the number of inflatons is large and the differential operator is high-dimensional.

Since we do not have fault-tolerant quantum computers yet, we cannot run our algorithm for the above problem now.
Instead, to see that our algorithm is promising, we numerically check that the overlap condition holds.
That is, for some cases with a small enough number of inflatons to be classically dealt with in the finite difference approach, we observe that the first eigenfunction has a considerable overlap with a simple test function of Gaussian shape.
This result implies that we can prepare the inputs of our algorithm for the stochastic inflation problem and thus run it on future fault-tolerant quantum computers.

This paper is organized as follows.
Sec.~\ref{sec:prel} is a preliminary one, where we present the basics of the finite difference method for calculating eigenvalues of differential operators and some quantum algorithms used as building blocks in this paper.
In Sec.~\ref{sec:ourAlgo}, we present our algorithm for estimating differential operator eigenvalues and evaluate its query complexity.
Sec.~\ref{sec: inflation} is devoted to the numerical demonstration concerning the application of our algorithm to the eigenvalue estimation in stochastic inflation.
After we present the basics of stochastic inflation, we see the considerable overlap between the eigenfunctions and the test functions in some cases.
This paper ends with the summary in Sec.~\ref{sec:sum}.

\section{Preliminary \label{sec:prel}}

\subsection{Notation}

$\mathbb{R}_+$ denotes the set of all non-negative real numbers: $\mathbb{R}_+%:=
\coloneqq\{x\in\mathbb{R} \ | \ x>0\}$.

For $n\in\mathbb{N}$, we define $[n]\coloneqq\{1,%...
\cdots,n\}$ and $[n]_0\coloneqq\{0,1,%...
\cdots,n-1\}$.

We label entries in a vector and rows and columns in a matrix with integers starting from 0.
That is, we write $\mathbf{v}\in\mathbb{C}^n$ as $\mathbf{v}=(v_0,%...
\cdots,v_{n-1})^T$ and $A\in\mathbb{C}^{m \times n}$ as $A=\begin{pmatrix} a_{0,0} & \cdots & a_{0,n-1} \\ \vdots & \ddots & \vdots \\ a_{m-1,0} & \cdots & a_{m-1,n-1} \end{pmatrix}$ entrywise.

For $n\in\mathbb{N}$, $I_n$ denotes the $n \times n$ identity matrix.
We may omit the subscript $n$ if there is no ambiguity.

For $\mathbf{v}\in\mathbb{C}^n$, $\|\mathbf{v}\|$ denotes its Euclidean norm.
For an (unnormalized) quantum state $\ket{\psi}$ on a multi-qubit system, $\|\ket{\psi}\|$ denotes the Euclidean norm of its state vector.
For $A\in\mathbb{C}^{m \times n}$, $\|A\|$ denotes its spectral norm.
$\|\mathbf{v}\|_\mathrm{max}$ (resp. $\|A\|_\mathrm{max}$) denotes the max norm of $\mathbf{v}$ (resp. $A$), which means the maximum of the absolute values of its entries.

For $\epsilon>0$, we say that $x^\prime\in\mathbb{R}$ is an $\epsilon$-approximation of $x\in\mathbb{R}$, if $|x^\prime-x|\le\epsilon$ holds.
We also say that $x^\prime$ is $\epsilon$-close to $x$.

If $A\in\mathbb{C}^{m \times n}$ has at most $s$ nonzero entries in each row and column, we say that $A$ is $s$-sparse and the sparsity of $A$ is $s$.

For a function $f:\Omega\rightarrow\mathbb{R}$ on $\Omega\in\mathbb{R}^d$ and $\boldsymbol{\alpha}=(\alpha_1,\cdots,\alpha_d)\in\mathbb{N}^d$, we define
\begin{equation}
    |\boldsymbol{\alpha}|\coloneqq\alpha_1+\cdots\alpha_d
\end{equation}
and
\begin{equation}
    D^{\boldsymbol{\alpha}}f \coloneqq \frac{\partial^{|\boldsymbol{\alpha}|}f}{\partial x_1^{\alpha_1}\cdots\partial x_d^{\alpha_d}}.
\end{equation}

We denote by $1_C$ the indicator function, which takes 1 if the condition $C$ is satisfied and 0 otherwise.

\subsection{Approximating eigenvalues of a differential operator by the finite difference method \label{sec:FDAppEigen}}

We now formulate the eigenvalue problem for a differential operator, focusing on the Sturm--Liouville type, to which the eigenvalue problem in stochastic inflation also boils down. 

\begin{problem}
Consider
\begin{equation}
    \mathcal{L}=-\sum_{i=1}^d \frac{\partial}{\partial x_i}\left(a_i\frac{\partial}{\partial x_i}\right)+a_0
    \label{eq:SturmLiou}
\end{equation}
a linear second-order self-adjoint elliptic differential operator on $\mathcal{D}%:=
\coloneqq(L,U)\times\cdots\times(L,U) \subset \mathbb{R}^d$, where $a_0,a_1,%...
\cdots,a_d:\overline{\mathcal{D}} \rightarrow \mathbb{R}_+$ are four-times continuously differentiable on $\overline{\mathcal{D}}$. 
We consider the eigenvalue problem for $\mathcal{L}$ with the Dirichlet boundary condition: we aim to find real numbers $\lambda_1 \le \lambda_2 \le \cdots$, each of which satisfies
\begin{equation}
\begin{cases}
    \mathcal{L}f_k(\mathbf{x})=\lambda_k f_k(\mathbf{x}) & ; \ \mathrm{for} \ \mathbf{x}\in\mathcal{D} \\
    f_k(\mathbf{x})=0 & ; \ \mathrm{for} \ \mathbf{x}\in\partial \mathcal{D}
\end{cases}
\end{equation}
for some function $f_k:\overline{\mathcal{D}}\rightarrow\mathbb{R}$, especially the smallest one $\lambda_1$.
\label{prob:eigen}
\end{problem}

We call each $\lambda_k$ an eigenvalue of $\mathcal{L}$ and $f_k$ the eigenfunction of $\mathcal{L}$ corresponding to $\lambda_k$.
We denote by $\mathcal{F}_\mathcal{L}$ the space of functions spanned by eigenfunctions of $\mathcal{L}$.

For this problem, we take a finite difference-based approach.
We start by setting grid points in $\overline{\mathcal{D}}$.
We take an integer $n_\mathrm{gr}\in\mathbb{N}$, which means the number of grid points in each of the $d$ dimensions, and denote by $N_\mathrm{gr}\coloneqq n_\mathrm{gr}^d$ the total number of grid points in $\mathcal{D}$.
We set the grid point interval $h\coloneqq(U-L)/(n_\mathrm{gr}+1)$.
We take the grid point set $\mathcal{D}_{n_\mathrm{gr}}$, which consists of the points $\mathbf{x}^\mathrm{gr}_\mathbf{j} \in \mathcal{D}$ labeled by $\mathbf{j}=(j_1,%...
\cdots,j_d)^T\in[n_\mathrm{gr}]_0^d$ and written as
\beae{
    \mathbf{x}^\mathrm{gr}_\mathbf{j}&=(x^\mathrm{gr}_{1,j_1},%...
    \cdots,x^\mathrm{gr}_{d,j_d})^T, \\
    x^\mathrm{gr}_{i,j_i} &= (j_i+1)h + L \ \mathrm{for \ each} \ i\in[d].
}
We label the point $\mathbf{x}^\mathrm{gr}_\mathbf{j}$ also by an integer $J=J(\mathbf{j})$, where
\begin{equation}
    J(\mathbf{j})\coloneqq\sum_{i=1}^d n_\mathrm{gr}^{i-1}j_i,
    \label{eq:IdConv}
\end{equation}
and denote it also by $\mathbf{x}^\mathrm{gr}_J$.
Hereafter, we sometimes label entries in a vector $v\in\mathbb{C}^{N_\mathrm{gr}}$ and rows and columns in a matrix $A\in\mathbb{C}^{N_\mathrm{gr} \times N_\mathrm{gr}}$ with $\mathbf{j}\in[n_\mathrm{gr}]_0^d$, where the one labeled by $\mathbf{j}$ corresponds to the $J(\mathbf{j})$-th one: e.g., $v_\mathbf{j}$ is the $J(\mathbf{j})$-th entry in the vector $(v_0,%...
\cdots,v_{N_\mathrm{gr}-1})^T$.

Using these grid points, we construct a finite difference approximation $L_{n_\mathrm{gr}}$ of $\mathcal{L}$.
That is, we want a $N_\mathrm{gr} \times N_\mathrm{gr}$ real symmetric matrix $L_{n_\mathrm{gr}}$ that satisfies
\begin{equation}
    \left|(L_{n_\mathrm{gr}} \mathbf{v}_{f,n_\mathrm{gr}})_\mathbf{j} - \mathcal{L}f(\mathbf{x}^\mathrm{gr}_\mathbf{j})\right| \xrightarrow{n_\mathrm{gr}\rightarrow\infty}0.
    \label{eq:FDAppL}
\end{equation}
for any $f\in \mathcal{F}_\mathcal{L}$ and any $\mathbf{j}\in[n_\mathrm{gr}]_0^d$.
Here, for $f:\overline{\mathcal{D}}\rightarrow\mathbb{R}$, we define $\mathbf{v}_{f,n_\mathrm{gr}}\in\mathbb{R}^{N_\mathrm{gr}}$ as
\begin{equation}
    \mathbf{v}_{f,n_\mathrm{gr}}\coloneqq(f(\mathbf{x}^\mathrm{gr}_0),%...
    \cdots,f(\mathbf{x}^\mathrm{gr}_{N_\mathrm{gr}-1}))^T.
\end{equation}
Specifically, we take $L_{n_\mathrm{gr}}$ as follows: for $\mathbf{j}_1,\mathbf{j}_2\in[n_\mathrm{gr}]_0^d$, its $(\mathbf{j}_1,\mathbf{j}_2)$-th entry is
\begin{align}
    & (L_{n_\mathrm{gr}})_{\mathbf{j}_1,\mathbf{j}_2} = \nonumber \\
    & 
    \begin{dcases}
    \sum_{i=1}^d \frac{a_i\left(\mathbf{x}^\mathrm{gr}_{\mathbf{j}_1}+\frac{h}{2}\mathbf{e}_i\right) + a_i\left(\mathbf{x}^\mathrm{gr}_{\mathbf{j}_1}-\frac{h}{2}\mathbf{e}_i\right)}{h^2} + a_0(\mathbf{x}^\mathrm{gr}_{\mathbf{j}_1}) \\
    \qquad \qquad \qquad \qquad \qquad \qquad \qquad \qquad \qquad \ \ \ \mathrm{if} \ \mathbf{j}_1=\mathbf{j}_2, \\
    -\frac{a_i\left(\mathbf{x}^\mathrm{gr}_{\mathbf{j}_1}\pm\frac{h}{2}\mathbf{e}_i\right)}{h^2} \qquad \mathrm{if} \ \mathbf{j}_2=\mathbf{j}_1 \pm \mathbf{e}_i \ \mathrm{for} \ \mathrm{some} \ i\in[d],\\
    0 \qquad \qquad \qquad \qquad \qquad \qquad \qquad \qquad \qquad \mathrm{otherwise,}
    \end{dcases}
    \label{eq:FDMat}
\end{align}
where $\mathbf{e}_i\in\mathbb{R}^d$ is a vector with all entries equal to 0 except for the $i$-th entry equal to 1.
In other words, for $\mathbf{j}\in[n_\mathrm{gr}]_0^d$,
\begin{align}
& (L_{n_\mathrm{gr}} \mathbf{v}_{f,n_\mathrm{gr}})_{\mathbf{j}} = \nonumber \\
&\ \sum_{i=1}^d \frac{1}{h^2}\left[-a_i\left(\mathbf{x}^\mathrm{gr}_{\mathbf{j}}+\frac{h}{2}\mathbf{e}_i\right)f(\mathbf{x}^\mathrm{gr}_{\mathbf{j}}+h\mathbf{e}_i)\right. \nonumber \\
    \ & \qquad \qquad +\left(a_i\left(\mathbf{x}^\mathrm{gr}_{\mathbf{j}}+\frac{h}{2}\mathbf{e}_i\right)+a_i\left(\mathbf{x}^\mathrm{gr}_{\mathbf{j}}-\frac{h}{2}\mathbf{e}_i\right)\right)f(\mathbf{x}^\mathrm{gr}_{\mathbf{j}}) \nonumber \\
    \ & \qquad \qquad   \left.-a_i\left(\mathbf{x}^\mathrm{gr}_{\mathbf{j}}-\frac{h}{2}\mathbf{e}_i\right)f(\mathbf{x}^\mathrm{gr}_{\mathbf{j}}-h\mathbf{e}_i)\right] \nonumber \\
    \ & + a_0(\mathbf{x})f(\mathbf{x}^\mathrm{gr}_{\mathbf{j}}).
    \label{eq:FDMatExp}
\end{align}
This $L_{n_\mathrm{gr}}$ is based on the following formula (see Sec.~6.2 in Ref.~\cite{larsson2003partial})
\begin{align}
    &\frac{\partial}{\partial x_i}\left(a_i\frac{\partial f}{\partial x_i}\right)(\mathbf{x}) = \nonumber \\
    &\ \frac{1}{h^2}\left[a_i\left(\mathbf{x}+\frac{h}{2}\mathbf{e}_i\right)f(\mathbf{x}+h\mathbf{e}_i)\right. \nonumber \\
    \ & \qquad -\left(a_i\left(\mathbf{x}+\frac{h}{2}\mathbf{e}_i\right)+a_i\left(\mathbf{x}-\frac{h}{2}\mathbf{e}_i\right)\right)f(\mathbf{x}) \nonumber \\
    \ & \qquad  + \left.a_i\left(\mathbf{x}-\frac{h}{2}\mathbf{e}_i\right)f(\mathbf{x}-h\mathbf{e}_i)\right] + O(h^2),
\end{align}
which holds for any sufficiently smooth function $f$.

Then, we expect that the eigenvalues of $L_{n_\mathrm{gr}}$ approximate those of $\mathcal{L}$ for sufficiently large $n_\mathrm{gr}$, and, in fact, this holds.
We have the following theorem by applying Theorem~5.1 in Ref.~\cite{Kuttler1970} to the current case. 

\begin{theorem}

For each $k\in\mathbb{N}$, there exist constants $C^k_\mathcal{L},D^k_\mathcal{L}\in\mathbb{R}$ determined only by $k$ and $\mathcal{L}$ such that, for any $n_\mathrm{gr}\in\mathbb{N}$, the following hold:
\begin{itemize}
\item
\begin{equation}
|\lambda^k_{n_\mathrm{gr}}-\lambda_k|\le \frac{C^k_\mathcal{L}}{n_\mathrm{gr}^2},
\end{equation}
where $\lambda^1_{n_\mathrm{gr}} \le \lambda^2_{n_\mathrm{gr}} \le \cdots$ are the eigenvalues of $L_{n_\mathrm{gr}}$.

\item 
For any eigenvector $\mathbf{v}^k_{n_\mathrm{gr}}$ of $L_{n_\mathrm{gr}}$ that corresponds to $\lambda^k_{n_\mathrm{gr}}$ and is normalized as $\|\mathbf{v}^k_{n_\mathrm{gr}}\|_{n_\mathrm{gr}}=1$, there exists an eigenfunction $f_k$ of $\mathcal{L}$ corresponding to $\lambda_k$ that satisfies
\begin{equation}
\|\mathbf{v}^k_{n_\mathrm{gr}}-\mathbf{v}_{f_k,n_\mathrm{gr}}\|_\mathrm{max}\le \frac{D^k_\mathcal{L}}{n_\mathrm{gr}^2}.
\end{equation}
Here, for $\mathbf{v}=(v_0,\cdots,v_{N_\mathrm{gr}-1})\in\mathbb{C}^{N_\mathrm{gr}}$, $\|\mathbf{v}\|_{n_\mathrm{gr}}\coloneqq h^d\sum_{J=0}^{N_\mathrm{gr}-1} |v_J|^2$.
\end{itemize}

\label{th:EigConv}
\end{theorem}

The proof is given in Appendix~\ref{sec:ProofEigConv}.

\subsection{Building-block quantum algorithms}

\subsubsection{Arithmetic circuits}

In this paper, we consider computation on the system with multiple quantum registers. 
We use the fixed-point binary representation for real numbers and, for each $x\in\mathbb{R}$, we denote by $\ket{x}$ the computational basis state on a quantum register where the bit string corresponds to the binary representation of $x$.
We assume that every register has a sufficient number of qubits and thus neglect errors by finite precision representation.

We can perform arithmetic operations on numbers represented on qubits.
For example, we can implement quantum circuits for four basic arithmetic operations: addition $O_{\mathrm{add}}:\ket{a}\ket{b}\ket{0}\mapsto\ket{a}\ket{b}\ket{a+b}$, subtraction $O_{\mathrm{sub}}:\ket{a}\ket{b}\ket{0}\mapsto\ket{a}\ket{b}\ket{a-b}$, multiplication $O_{\mathrm{mul}}:\ket{a}\ket{b}\ket{0}\mapsto\ket{a}\ket{b}\ket{ab}$, and division $O_{\mathrm{div}}:\ket{a}\ket{b}\ket{0}\ket{0}\mapsto\ket{a}\ket{b}\ket{q}\ket{r}$, where $a,b\in\mathbb{Z}$ and $q$ and $r$ are the quotient and remainder of $a/b$, respectively.
For concrete implementations, see Ref.~\cite{MunosCoreas2022} and the references therein.
In the finite precision binary representation, these operations are immediately extended to those for real numbers.
Moreover, using the above circuits, we can also construct a circuit to compute an elementary function $f$ such as an exponential as $O_f:\ket{x}\ket{0}\mapsto\ket{x}\ket{f(x)}$; see \cite{haner2018optimizing} for the detail of the construction.
Hereafter, we collectively call these circuits arithmetic circuits.
The gate cost of such a circuit is polynomial with respect to the number of digits in binary representation and not related to the parameters in the currently considered problem such as $\epsilon$, as long as a sufficient number of digits are used.
Thus, we regard their contributions as constant in the complexity evaluation conducted later.
That is, in such an evaluation, we regard them as a unit and count the number of queries to them and circuits constructed by their combinations (e.g., $O_{a_i}$ in Eq.~\eqref{eq:Oai}).

\subsubsection{Representation of a vector as a quantum state \label{sec:VectorEncode}}

For a vector $\mathbf{v}=(v_0,%...
\cdots,v_{n-1})\in\mathbb{R}^n$, we consider the two ways to encode it into a quantum state.
The first one is amplitude encoding: we define 
\begin{equation}
    \ket{\mathbf{v}}\coloneqq \frac{1}{\sqrt{\sum_{i=0}^{n-1}v_i^2}}\sum_{i=0}^{n-1} v_i\ket{i}.
\end{equation}
We will later consider quantum states like this for vectors like $\mathbf{v}_{f,n_\mathrm{gr}}$, which consists of the values of functions at grid points.
For a function written by some explicit formula, there are various methods to generate a amplitude-encoding state~\cite{grover2002creating,Sanders2019,zoufal2019quantum,Holmes2020,rattew2022preparing,mcardle2022quantum,MarinSanchez2023,Moosa_2024}.

The second way is binary encoding: we define
\begin{equation}
    \ket{\mathbf{v}}_\mathrm{bin}\coloneqq \ket{v_0}\cdots\ket{v_{n-1}}.
\end{equation}

\subsubsection{Block-encoding}

Block-encoding means embedding a general matrix into the upper-left block of a unitary matrix.

\begin{definition}[Ref.~\cite{Gilyen2019}, Definition 24]
Let $n,a\in\mathbb{N}$, $A\in\mathbb{C}^{2^n\times2^n}$,  $\epsilon\in\mathbb{R}_+$, and $\alpha\ge\|A\|$.
We say that a unitary $U$ on a $(n+a)$-qubit system is an $(\alpha,a,\epsilon)$-block-encoding of $A$, if
\begin{equation}
    \left\|A-\alpha(\bra{0}^{\otimes a}\otimes I_{2^n})U(\ket{0}^{\otimes a}\otimes I_{2^n})\right\|\le\epsilon
\end{equation}
\end{definition}

We can efficiently construct a block-encoding of a matrix $A$ if we have sparse-access to $A$, that is, if $A$ is sparse and we can query quantum circuits that output positions and values of nonzero entries in $A$.

\begin{theorem}[Ref.~\cite{Gilyen2019}, Lemma 48 in the full version, modified]

Let $A=(A_{ij})\in\mathbb{C}^{2^n \times 2^n}$ be a $s$-sparse matrix.
Suppose that we have accesses to the oracles $O^A_\mathrm{row}$ and $O^A_\mathrm{col}$ that act on a two register system as
\begin{equation}
    O^A_\mathrm{row}\ket{i}\ket{k}=\ket{i}\ket{r_{ik}},O^A_\mathrm{col}\ket{k}\ket{i}=\ket{c_{ki}}\ket{i}
    \label{eq:OARowCol}
\end{equation}
for any $i\in[2^n]_0$ and $k\in[s]$, where $r_{ik}$ (resp. $c_{ki}$) is the index of the $k$th nonzero entry in the $i$th row (resp. column) in $A$, or $i+2^n$ if there are less than $k$ nonzero entries.
Besides, suppose that we have accesses to the oracle $O^A_\mathrm{ent}$ that acts on a three register system as
\begin{equation}
    O^A_\mathrm{ent}\ket{i}\ket{j}\ket{0}=\ket{i}\ket{j}\ket{A_{ij}}.
    \label{eq:OAEnt}
\end{equation}
Then, for any $\epsilon\in\mathbb{R}_+$, there exists a $(s\|A\|_\mathrm{max},n+3,\epsilon)$-block-encoding of $A$, in which $O^A_\mathrm{row}$ and $O^A_\mathrm{col}$ are each queried once, $O^A_\mathrm{ent}$ is queried twice, additional $O\left(n+\log^{5/2}\left(\frac{s^2\|A\|_\mathrm{max}}{\epsilon}\right)\right)$ 1- and 2-qubit gates are used, and $O\left(\log^{5/2}\left(\frac{s^2\|A\|_\mathrm{max}}{\epsilon}\right)\right)$ ancilla qubits are used.
\label{th:BlEncSp}
\end{theorem}

\subsubsection{Estimation of the smallest eigenvalue of a hermitian}

Given a block-encoding of a matrix $A$, we can use a technique called \ac{QSVT}~\cite{Gilyen2019,Martyn2021} to construct a block-encoding of another matrix $\tilde{A}$ obtained via transforming the singular values of $A$ by a polynomial $f$ satisfying some conditions.
For a hermitian $H$, this corresponds to transforming its eigenvalues, and various operations concerning the eigenvalues are possible via \ac{QSVT}.
Ref.~\cite{Lin2020nearoptimalground} utilized this scheme to construct a quantum algorithm to estimate the smallest eigenvalue of a hermitian.

\begin{theorem}[Ref.~\cite{Lin2020nearoptimalground}, Theorem 8]

Let $H$ be a $2^n \times 2^n$ hermitian with the smallest eigenvalue $\lambda_1$ and the corresponding eigenvector $\ket{\psi_1}$.
Suppose that we have access to $(\alpha,a,0)$-block-encoding $U_H$ of $H$.
Also suppose that we have access to the unitary $U_{\ket{\phi_1}}$ on a $n$-qubit system to generate a state $\ket{\phi_1}$, that is, $U_{\ket{\phi_1}}\ket{0}=\ket{\phi_1}$, where $\left|\left\langle\phi_1\middle| \psi_1\right\rangle\right|\ge\gamma$ holds with some $\gamma>0$.
Then, for any $\delta\in(0,1)$ and $\epsilon>0$, there exists a quantum algorithm that outputs an $\epsilon$-approximation of $\lambda_1$ with probability at least $1-\delta$, using $U_H$
\begin{equation}
    O\left(\frac{\alpha}{\gamma \epsilon}\log\left(\frac{\alpha}{ \epsilon}\right)\log\left(\frac{1}{ \gamma}\right)\log\left(\frac{\log(\alpha/\epsilon)}{ \delta}\right)\right)
    \label{eq:GEEUHQuery}
\end{equation}
times, $U_{\ket{\phi_0}}$
\begin{equation}
    O\left(\frac{1}{\gamma}\log\left(\frac{\alpha}{ \epsilon}\right)\log\left(\frac{\log(\alpha/\epsilon)}{ \delta}\right)\right)
    \label{eq:GEEUiniQuery}
\end{equation}
times,
\begin{equation}
    O\left(\frac{a\alpha}{\gamma \epsilon}\log\left(\frac{\alpha}{ \epsilon}\right)\log\left(\frac{1}{ \gamma}\right)\log\left(\frac{\log(\alpha/\epsilon)}{ \delta}\right)\right)
    \label{eq:GEEGateNum}
\end{equation}
additional 1- and 2-qubit gates, and
\begin{equation}
    O\left(n+a+\log\left(\frac{1}{ \gamma}\right)\right)
    \label{eq:GEEQubitNum}
\end{equation}
qubits.

    \label{th:GEE}
\end{theorem}

We can modify the quantum algorithm in Ref.~\cite{Lin2020nearoptimalground}, for which the availability of an exact block-encoding of $H$ is assumed, to an algorithm built upon sparse-access to $H$ and thus an approximate block-encoding of it.

\begin{corollary}
    Let the assumptions in Theorem~\ref{th:GEE} holds, except that on the $(\alpha,a,0)$-block-encoding of $H$ replaced with the following: $H$ is $s$-sparse we have access to $O^H_\mathrm{row}$, $O^H_\mathrm{col}$ and $O^H_\mathrm{ent}$ that act as Eqs.~(\ref{eq:OARowCol}) and (\ref{eq:OAEnt}) with $A=H$.
    Then, for any $\delta\in(0,1)$ and $\epsilon>0$, there exists a quantum algorithm $\proc{EstEig}(H,\epsilon,\delta)$ that outputs an $\epsilon$-approximation of $\lambda_1$ with probability at least $1-\delta$, using $O^H_\mathrm{row}$, $O^H_\mathrm{col}$ and $O^H_\mathrm{ent}$
    \begin{align}
        &O\left(\frac{s\|H\|_\mathrm{max}}{\gamma \epsilon}\log\left(\frac{s\|H\|_\mathrm{max}}{ \epsilon}\right)\log\left(\frac{1}{ \gamma}\right)\right. \nonumber \\
        &\qquad\left.\times\log\left(\frac{\log(s\|H\|_\mathrm{max}/\epsilon)}{ \delta}\right)\right)
        \label{eq:GEEUHQuerySpOra}
    \end{align}
    times, $U_{\ket{\phi_0}}$
    \begin{equation}
        O\left(\frac{1}{\gamma}\log\left(\frac{s\|H\|_\mathrm{max}}{ \epsilon}\right)\log\left(\frac{\log(s\|H\|_\mathrm{max}/\epsilon)}{ \delta}\right)\right)
        \label{eq:GEEUiniQuerySpOra}
    \end{equation}
    times,
    \begin{align}
        &O\left(\frac{s\|H\|_\mathrm{max}}{\gamma \epsilon}\log\left(\frac{s\|H\|_\mathrm{max}}{ \epsilon}\right)\log\left(\frac{1}{ \gamma}\right)\right.\nonumber \\
        &\qquad \times \log\left(\frac{\log(s\|H\|_\mathrm{max}/\epsilon)}{ \delta}\right) \nonumber \\
        &\qquad\left. \times\left(n+\log^{5/2}\left(\frac{s^3\|H\|_\mathrm{max}^2\log^2(1/\gamma)}{\gamma^2\epsilon^2}\right)\right)\right)
        \label{eq:GEEGateNumSpOra}
    \end{align}
    additional 1- and 2-qubit gates, and
    \begin{equation}
        O\left(n+\log\left(\frac{1}{ \gamma}\right)+\log^{5/2}\left(\frac{s^3\|H\|_\mathrm{max}^2\log^2(1/\gamma)}{\gamma^2\epsilon^2}\right)\right)
        \label{eq:GEEQubitNumSpOra}
    \end{equation}
    qubits.
    \label{co:GEESparse}
\end{corollary}

\begin{proof}
    See Sec.~\ref{sec:ProofGEESparse}.
\end{proof}

\section{Improved quantum algorithm for eigenvalue estimation for differential operators \label{sec:ourAlgo}}

Now, let us present our quantum algorithm for estimating the smallest eigenvalue of a differential operator $\mathcal{L}$ in the form of Eq.~\eqref{eq:SturmLiou}.
We begin by making assumptions about access to quantum circuits used in our algorithms.
The first one is the circuits to compute the coefficient $a_i$ in $\mathcal{L}$.

\begin{assumption}
    For each $i\in\{0,1,\ldots,d\}$, we have access to a quantum circuit $O_{a_i}$ that acts as
    \begin{equation}
        O_{a_i} \ket{\mathbf{x}}_\mathrm{bin}\ket{0}= \ket{\mathbf{x}}_\mathrm{bin}\ket{a_i(\mathbf{x})}
        \label{eq:Oai}
    \end{equation}
    for any $\mathbf{x}\in\mathcal{D}$.
    \label{ass:Oai}
\end{assumption}

As long as $a_i$ is given by some explicit formula, which is the case in the stochastic inflation case considered below, $O_{a_i}$ can be implemented with arithmetic circuits.

The second one is the circuit to generate the quantum state that encodes the test function, which is chosen in advance and assumed to overlap well with the first eigenfunction of $\mathcal{L}$.

\begin{assumption}
    Let $\gamma\in(0,1)$.
    For any $n_\mathrm{gr}\in\mathbb{N}$, we have access to a quantum circuit $O_{\tilde{f}_1,n_\mathrm{gr}}$ that acts as
    \begin{equation}
        O_{\tilde{f}_1,n_\mathrm{gr}} \ket{0}= \ket{\mathbf{v}_{\tilde{f}_1,n_\mathrm{gr}}},
    \end{equation}
    where $\tilde{f}_1:\mathcal{D}\rightarrow\mathbb{R}$ is a function on $\mathcal{D}$ and satisfies
    \begin{equation}
        |\braket{\mathbf{v}_{\tilde{f}_1,n_\mathrm{gr}} | \mathbf{v}_{f_1,n_\mathrm{gr}}}|\ge\gamma.
        \label{eq:f1f1til}
    \end{equation}
    \label{ass:Of1}
\end{assumption}

As mentioned in Sec.~\ref{sec:VectorEncode}, for the function $\tilde{f}_1$ written by an explicit formula, there are various proposals on the way to generate the state $\ket{\mathbf{v}_{\tilde{f}_1,n_\mathrm{gr}}}$, and we now assume that some of them are available.

Then, we present the main theorem on our quantum algorithm.

\begin{theorem}
    Consider Problem \ref{prob:eigen} under Assumptions \ref{ass:Oai} and \ref{ass:Of1}.
    Let $\epsilon\in\mathbb{R}$ and $\delta\in(0,1)$.
    Then, there exists a quantum algorithm that outputs an $\epsilon$-approximation of $\lambda_1$ with probability at least $1-\delta$, making
    \begin{equation}
    O\left(\frac{d\xi}{\gamma}\log \xi \log\left(\frac{1}{ \gamma}\right)\log\left(\frac{\log\xi}{\delta}\right)\right)
        \label{eq:EigEstCompOai}
    \end{equation}
    uses of $O_{a_0},%...
    \cdots,O_{a_d}$ and arithmetic circuits and
    \begin{equation}
        O\left(\frac{1}{\gamma}\log \xi \log\left(\frac{\log\xi}{\delta}\right)\right)
        \label{eq:EigEstCompOf1}
    \end{equation}
    uses of $O_{\tilde{f}_1,n_\mathrm{gr}}$ with
    \begin{equation}
        n_\mathrm{gr}=\left\lceil\max\left\{\sqrt{\frac{2C^1_\mathcal{L}}{\epsilon}},\sqrt{\frac{2D^1_\mathcal{L}}{1-\eta(\gamma)}}(U-L)^{d/4}\right\}\right\rceil.
        \label{eq:ngrForEps}
    \end{equation}
    Here,
    \begin{align}
        &\xi\coloneqq \nonumber \\
        & \ \frac{d}{\epsilon}\left(\frac{d a_\mathrm{max}}{(U-L)^2} \times \max\left\{\frac{C^1_\mathcal{L}}{\epsilon},\frac{D^1_\mathcal{L}(U-L)^{\frac{d}{2}}}{1-\eta(\gamma)}\right\}+ a_{0,\mathrm{max}}\right),
    \end{align}
    \begin{equation}
    a_{\mathrm{max}}\coloneqq\max_{\substack{i\in[d] \\ \mathbf{x}\in\overline{\mathcal{D}}}} a_i(\mathbf{x}) \qc a_{0,\mathrm{max}}\coloneqq\max_{ \mathbf{x}\in\overline{\mathcal{D}}} a_0(\mathbf{x})
    \end{equation}
    and, for $x\in(0,1)$,
    \begin{equation}
        \eta(x)\coloneqq\frac{1}{2}\left(x^2+\sqrt{4-5x^2+x^4}\right).
    \end{equation}
    \label{th:main}
\end{theorem}

\begin{proof}

    Here, we just present the quantum algorithm and postpone the rest of the proof, that is, presenting how to construct $O^{L_{n_\mathrm{gr}}}_\mathrm{row}$, $O^{L_{n_\mathrm{gr}}}_\mathrm{col}$ and $O^{L_{n_\mathrm{gr}}}_\mathrm{ent}$, estimating the accuracy of the output, and the query complexity estimation, to Appendix~\ref{sec:ProofMain}.

    \begin{algorithm}[H]
    \begin{algorithmic}[1]
    \REQUIRE Accuracy $\epsilon\in\mathbb{R}$, success probability $1-\delta\in(0,1)$.

    \STATE Set $n_\mathrm{gr}$ as Eq. \eqref{eq:ngrForEps}.

    \STATE Construct the oracles $O^{L_{n_\mathrm{gr}}}_\mathrm{row}$, $O^{L_{n_\mathrm{gr}}}_\mathrm{col}$ and $O^{L_{n_\mathrm{gr}}}_\mathrm{ent}$.

    \STATE Run $\proc{EstEig}\left(L_{n_\mathrm{gr}},\frac{\epsilon}{2},\delta\right)$ to get an $\frac{\epsilon}{2}$-approximation $\tilde{\Lambda}$ of the smallest eigenvalue of $L_{n_\mathrm{gr}}$.

    \STATE Output $\tilde{\Lambda}$.
    
    \caption{Estimation of the first eigenvalue $\lambda_1$ of $\mathcal{L}$}
    \label{alg:EigEstim}
    \end{algorithmic}
\end{algorithm}
    
\end{proof}

Assuming that we can take a good function $\tilde{f}_1$ with $\gamma=\Theta(1)$, we take only the leading part with respect to $\frac{1}{\epsilon}$ to simplify Eq.~\eqref{eq:EigEstCompOai} as
\begin{equation}
    \widetilde{O}\left(\frac{d^3 a_\mathrm{max} C^1_\mathcal{L}}{(U-L)^2\epsilon^2}\right).
    \label{eq:EigEstCompOaiSimple}
\end{equation}
This evaluation can be decomposed into the following contributions.
For $L_{n_\mathrm{gr}}$, the sparsity is $s=O(d)$, and the max norm is, as seen in Appendix~\ref{sec:ProofMain}, $\|L_{n_\mathrm{gr}}\|_{\rm max}=O(d a_\mathrm{max} C^1_\mathcal{L}/(U-L)^2\epsilon)$ for small $\epsilon$.
Plugging these into Eq.~\eqref{eq:GEEUHQuerySpOra} and noting that $O^{L_{n_\mathrm{gr}}}_\mathrm{row}$, $O^{L_{n_\mathrm{gr}}}_\mathrm{col}$ and $O^{L_{n_\mathrm{gr}}}_\mathrm{ent}$ are implemented by $O(d)$ uses of $O_{a_0}, \cdots,O_{a_d}$ and arithmetic circuits as shown in Appendix~\ref{sec:ProofMain}, we reach Eq.~\eqref{eq:EigEstCompOaiSimple}.

Let us comment on the improvement of the complexity in our algorithm compared to the previous algorithm in Ref.~\cite{Szkopek2005}.
That algorithm also takes a kind of finite difference approximation $L^\prime$ for $\mathcal{L}$, which is different from Eq. \eqref{eq:FDMat}, but the estimation of the eigenvalue of $L^\prime$ is done by the Abrams-Lloyd algorithm~\cite{Abrams1999}.
It is the combination of the Hamiltonian simulation with $H=L^\prime$, which means applying the time evolution operator $\exp(-iL^\prime t)$ to quantum states, and the quantum Fourier transform.
Ref.~\cite{Szkopek2005} did not present any complexity upper bound of their algorithm that shows the scaling on all the parameters collectively but only stated the scaling of the complexity on the accuracy $\epsilon$.
According to it, their algorithm applied to the current problem has the complexity of order $\widetilde{O}(1/\epsilon^3)$.
This complexity is affected by that of the Hamiltonian simulation method adopted in Ref.~\cite{Szkopek2005}, which is a kind of the Suzuki-Trotter decomposition and not the state-of-the-art method based on \ac{QSVT}.
Compared to this, our algorithm uses the eigenvalue estimation method in Ref.~\cite{Lin2020nearoptimalground}, which does not involve the Hamiltonian simulation but is a binary search algorithm utilizing the \ac{QSVT}-based eigenvalue thresholding.
Thus, our algorithm improves the complexity as Eq.~\eqref{eq:EigEstCompOaiSimple}.

\section{Application to estimating the decay rate of the perturbation distribution tail in stochastic inflation}\label{sec: inflation}

We exemplify an application of the quantum algorithm we showed to \emph{cosmic inflation} as physical interest.
Cosmic inflation~\cite{Starobinsky:1980te,Sato:1981qmu,Guth:1980zm,Linde:1981mu,Albrecht:1982wi,Linde:1983gd} is the hypothetical phase of accelerated expansion in the early universe.
Not only can it make our universe statistically homogeneous and isotropic from a global perspective, but inflation can also bring primordial fluctuations in energy density or the spacetime metric from the quantum vacuum fluctuation.
Such primordial perturbations can be seeds of the current cosmological structures such as galaxies and clusters, and these observations support the existence of inflation (see, e.g., Ref.~\cite{Planck:2018jri}).

The fluctuating dynamics of inflation is often described in the \emph{stochastic formalism of inflation}, also known as \emph{stochastic inflation} (see Refs.~\cite{Starobinsky:1982ee,Starobinsky:1986fx,Nambu:1987ef,Nambu:1988je,Kandrup:1988sc,Nakao:1988yi,Nambu:1989uf,Mollerach:1990zf,Linde:1993xx,Starobinsky:1994bd} for the first works and also Ref.~\cite{Cruces:2022imf} for a recent review).
It is understood as an effective theory of fields coarse-grained on a superHubble scale.
We suppose that inflation is driven by $d$ canonical real scalar fields $\bm{\phi}=(\phi_1,\phi_2,\cdots\phi_d)$ called \emph{inflatons}.
Then, in stochastic inflation, the inflaton fields at each spatial point behave as almost independent stochastic processes.
In the so-called slow-roll limit, the stochastic differential equation which the inflatons follow at each spatial point is exhibited as (see, e.g., Ref.~\cite{Vennin:2015hra})
\bme{\label{eq: EoM}
    \!\!\!\!\!\!
    \dd{\bm\phi}\!(N)\!=\!-\Mpl^2\frac{\nabla_{\bm{\phi}}v(\bm{\phi}(N))}{v(\bm{\phi}(N))}\dd N+\sqrt{2v(\bm{\phi}(N))}\dd{W(N)},
}
where $\Mpl$ is the reduced Planck mass, $v=V/(24\pi^2\Mpl^4):\bbR^d\to\bbR$ is the reduced potential of the inflatons, the e-folding number $N$ is the time variable normalised by the Hubble parameter, and $W(N)$ is the $d$-dimensional Wiener process independent over the Hubble distance.\footnote{The stochastic term should be understood as the It\^o integral~\cite{Tokuda:2017fdh,Tokuda:2018eqs}. However, the It\^o integral breaks the covariance under the general coordinate transformation on the inflatons' target manifold if it is not Euclidean $\bbR^d$ but a more general one $\calM$~\cite{Pinol:2018euk}. In such a case, the inflatons' derivative $\dd{\phi_i}$ should be replaced by the \emph{It\^o-covariant} one. See Ref.~\cite{Pinol:2020cdp} for the details.}
We omitted the spatial label $\bfx$ as hereafter we only deal with the one-point dynamics in this paper.

Due to their stochastic behaviour, inflation continues for different times at each spatial point even if the inflaton fields have the same initial value $\bm{\phi}_0$. 
Supposing that inflation happens in a certain subset of the target manifold $\Omega\subset\bbR^I$ and ends at its boundary $\partial\Omega$, the inflation duration is given by the first passage time denoted $\calN(\bm{\phi}_0)$ from $\bm{\phi}_0$ to $\partial\Omega$.
According to the $\delta N$ formalism~\cite{Starobinsky:1985ibc,Salopek:1990jq,Sasaki:1995aw,Sasaki:1998ug,Wands:2000dp,Lyth:2004gb,Lyth:2005fi}, the fluctuation in this first passage time is understood as the conserved curvature perturbation $\zeta$ (fluctuation in the spatial curvature), which is converted to the energy density contrast in the later universe.
Its correlation functions over different spatial points, though we do not explicitly calculate them in this paper, are calculated by the \ac{PDF} of $\calN(\bm{\phi}_0)$ and its dependence on $\bm{\phi}_0$ known as the \emph{stochastic-$\delta\calN$ technique}~\cite{Fujita:2013cna,Fujita:2014tja,Vennin:2015hra,Ando:2020fjm,Tada:2021zzj,Animali:2024jiz}.
Hence, the problem of interest reduces to solving the \ac{PDF} of $\calN$ from each field value $\bm{\phi}\in\Omega$, omitting the subscript $0$ here and hereafter for simplicity.

The \ac{PDF} $P(\bm{\phi}\mid N)$ of the inflatons $\bm{\phi}$ at a certain time $N$ follows the \ac{FP} equation equivalent to the original stochastic differential equation~\eqref{eq: EoM} as
\bae{
    \partial_NP(\bm{\phi}\mid N)&=\calL_\FP P(\bm{\phi}\mid N) \nonumber \\
    &\bmbe{\coloneqq\Mpl^2\left[\sum_i\partial_{\phi_i}\pqty{\frac{v_i(\bm{\phi})}{v(\bm{\phi})}P(\bm{\phi}\mid N)} \right. \\
    \left.+\sum_i\partial_{\phi_i}^2\pqty{v(\bm{\phi})P(\bm{\phi}\mid N)}\right],}
}
where $v_i=\partial_{\phi_i}v$.
On the other hand, the \ac{PDF} of the first passage time $P_\FPT(\calN\mid\bm{\phi})$ is known to follow the adjoint one~\cite{Vennin:2015hra}:
\bae{\label{eq: adjoint FP}
    \partial_\calN P_\FPT(\calN\mid\bm{\phi})=\calL_\FP^\dagger P_\FPT(\calN\mid\bm{\phi}),
}
with the adjoint \ac{FP} operator
\bae{
    \frac{1}{\Mpl^2}\calL_\FP^\dagger=-\sum_i\frac{v_i}{v}\partial_{\phi_i}+v\sum_i\partial_{\phi_i}^2,
}
associated with the inner product,
\bae{
    \braket{f(\bm{\phi})|g(\bm{\phi})}\coloneqq\int\dd{\bm{\phi}}f(\bm{\phi})g(\bm{\phi}),
}
that is, $\braket{f(\bm{\phi})|\calL_\FP g(\bm{\phi})}=\braket{\calL_\FP^\dagger f(\bm{\phi})|g(\bm{\phi})}$.
The eigenvalues of $\calL_\FP^\dagger$ are negative in the set of functions valid as a PDF of the first passage time \cite{Ezquiaga:2019ftu}, and we hereafter consider the eigenvalues of $-\calL_\FP^\dagger$, which are positive.
Though $\calL_\FP^\dagger$ is not Hermite (not self-adjoint; $\calL_\FP^\dagger\neq\calL_\FP$), one can Hermitise it by defining the following Hermitian operator~\cite{Ezquiaga:2019ftu},
\bae{
    \widetilde{\calL_\FP}\coloneqq w^{1/2}(\bm{\phi})\calL_\FP^\dagger w^{-1/2}(\bm{\phi}),
}
where
\bae{
    w(\bm{\phi})=\ee^{1/v(\bm{\phi})}/v(\bm{\phi}).
}
They share the same eigenvalues, while the corresponding eigenfunctions are different by the $w^{-1/2}$ factor:
\bae{\label{eq: eigenvalue equation}
    -\widetilde{\calL_\FP}\Psi_n(\bm{\phi})=\Lambda_n\Psi_n(\bm{\phi}) \, \Leftrightarrow \, -\calL_\FP^\dagger\tilde{\Psi}_n(\bm{\phi})=\Lambda_n\tilde{\Psi}_n(\bm{\phi}),
}
where
\bae{
    \tilde{\Psi}_n(\bm{\phi})=w^{-1/2}(\bm{\phi})\Psi_n(\bm{\phi}).
}
By some calculations, we see that
\bae{
    &-\frac{\widetilde{\calL_\FP}}{\Mpl^2}= \nonumber \\
    & \quad \sum_i\left[-\partial_{\phi_i}\left(v\partial_{\phi_i}\right)-\frac{2v^2(1+v)v_{ii}-(1+4v+v^2){v_i}^2}{4v^3}\right],
}
where $v_{ii}\coloneqq\partial_{\phi_i}^2 v$, and thus $\widetilde{\calL_\FP}$ takes the form of Eq.~\eqref{eq:SturmLiou}.

The eigensystem of the adjoint \ac{FP} operator has a physically interesting implication.
The adjoint \ac{FP} equation~\eqref{eq: adjoint FP} can formally be solved as
\bae{
    P_\FPT(\calN\mid\bm{\phi})=\exp[(\calN-\epsilon)\calL_\FP^\dagger]P_\FPT(\calN=\epsilon\mid\bm{\phi}),
}
with a certain positive parameter $\epsilon$.
Expanding $P_\FPT(\calN=\epsilon\mid\bm{\phi})$ by the eigensystem as\footnote{If the target space $\Omega$ is non-compact, the spectrum of the eigensystem is not discrete in general.}
\bae{
    P_\FPT(\calN=\epsilon\mid\bm{\phi})=\sum_n\alpha_n^{(\epsilon)}\tilde{\Psi}_n(\bm{\phi}),
}
it leads to the solution
\bae{
    P_\FPT(\calN\mid\bm{\phi})=\sum_n\alpha_n^{(\epsilon)}\tilde{\Psi}_n(\bm{\phi})\ee^{-\Lambda_n(\calN-\epsilon)}.
}
One may take the limit $\epsilon\to0$ as
\bae{\label{eq: expansion of PFPT}
    P_\FPT(\calN\mid\bm{\phi})=\sum_n\alpha_n\tilde{\Psi}_n(\bm{\phi})\ee^{-\Lambda_n\calN},
}
where $\alpha_n=\lim_{\epsilon\to0}\alpha_n^{(\epsilon)}$.
The functional form of Eq.~\eqref{eq: expansion of PFPT} suggests that the \ac{PDF} of $\calN$ (and hence the primordial perturbation) has an exponentially heavy tail~\cite{Pattison:2017mbe,Ezquiaga:2019ftu,Figueroa:2020jkf} in contrast to the na\"ive expectation that the physical perturbations are well described by the Gaussian distribution.
In particular, the first eigenvalue, especially if it is of order unity, can exhibit a significant effect on the large-$\calN$ (and hence large perturbation) probability and drastically change the abundance of astrophysical objects.
This is why we want the method to calculate the first eigenvalue and the corresponding eigenfunction of the adjoint \ac{FP} operator (or equivalently the Hermite \ac{FP} operator $\widetilde{\calL_\FP}$).
However, this is a challenging computational task in classical computing, especially when the number of fields $d$ is much larger than 1, and our quantum eigenvalue-finding algorithm proposed in Sec. \ref{sec:ourAlgo} may be beneficial.

Let us briefly discuss the behaviour of the Hermite FP operator before moving to specific examples.
First of all, because $P_\FPT(\calN>0\mid\bm{\phi}\in\partial\Omega)=0$ by definition of $\calN$ and $\partial\Omega$,\footnote{$P_\FPT(\calN\mid\bm{\phi}\in\partial\Omega)=\delta(\calN)$ according to the conservation of the probability.}
the boundary condition for the eigenfunctions is given by
\bae{
    \tilde{\Psi}_n(\bm{\phi}\in\partial\Omega)=0 \, \Leftrightarrow \,
    \Psi_n(\bm{\phi}\in\partial\Omega)=0.
}
In a single-field model, $d=1$ (and $\phi$ denotes $\phi_1$ for brevity), the eigenvalue equation~\eqref{eq: eigenvalue equation} reduces to
\bae{
    \pqty{\partial_\phi^2+\frac{v'}{v}\partial_\phi+\omega_n^2}\Psi_n(\phi)=0,
}
with
\bae{\label{eq: omega2}
    \omega_n^2\simeq-\frac{1-2\Lambda_n\calP_\zeta-2\eta_\sto}{2\Mpl^2v\calP_\zeta}.
}
Here, we define $\calP_\zeta\coloneqq2v^3/({v'}^2\Mpl^2)>0$ and $\eta_\sto\coloneqq v^2v''/{v'}^2$, and we suppose $0<v\ll1$ in order for the inflation energy scale to be well below the Planck scale.
$\calP_\zeta$ is related to the amplitude of the curvature perturbation $\zeta$ in the perturbative evaluation.
$\eta_\sto$ is called \emph{stochasticity parameter} which indicates the magnitude of the stochastic correction to the perturbative evaluation~\cite{Vennin:2015hra}.
It is decomposed as $\eta_\sto=\eta_V\calP_\zeta/2$ where $\eta_V\coloneqq \Mpl^2v''/v^2$ is known as the second slow-roll parameter and supposed to be small for single-field slow-roll inflation.
Therefore, if the dynamics is well perturbative as $\calP_\zeta\ll1$, $\eta_\sto$ is also small, and then for $\Lambda_n\sim O(1)$, $\omega_n^2$ becomes negative, which means that the eigenfunction is not normalizable.
That is, unless the perturbativity is broken as $\calP_\zeta\gtrsim1$ in some region of the target space, $\widetilde{\calL_\FP}$ have no order-unity eigenvalue, and thus there is no interesting physics caused by the large-$\calN$ probability. 
Though this condition can be relaxed in multi-field cases, the above discussion motivates us to seek the situation with $\calP_\zeta\gtrsim1$.
Then, we below focus on models with a flat point $v'=0$ at which $\calP_\zeta$ can be divergent.

Although we would like to run our quantum algorithm for concrete problems including those with many fields, it is impossible today since there is no large-scale fault-tolerant quantum computer.
Instead, to see if our method is promising for the eigenvalue problem in stochastic inflation, we take some classically (or even analytically) tractable cases with a few fields and check that we can take a test function $\tilde{f}_1$ overlapping with the true first eigenfunction well.
With this condition satisfied, it is expected that our quantum algorithm finds the first eigenvalue efficiently.
Specifically, expecting that the first eigenfunction has a simple functional shape with a single bump and no node in many cases, we take Gaussian test functions, which seem to work at least in the examined cases below.

\subsection{Quantum well toy-model \label{sec:qwell}}

\begin{figure*}
	\centering
	\begin{tabular}{c}
		\begin{minipage}{0.5\hsize}
			\centering
			\includegraphics[width=0.95\hsize]{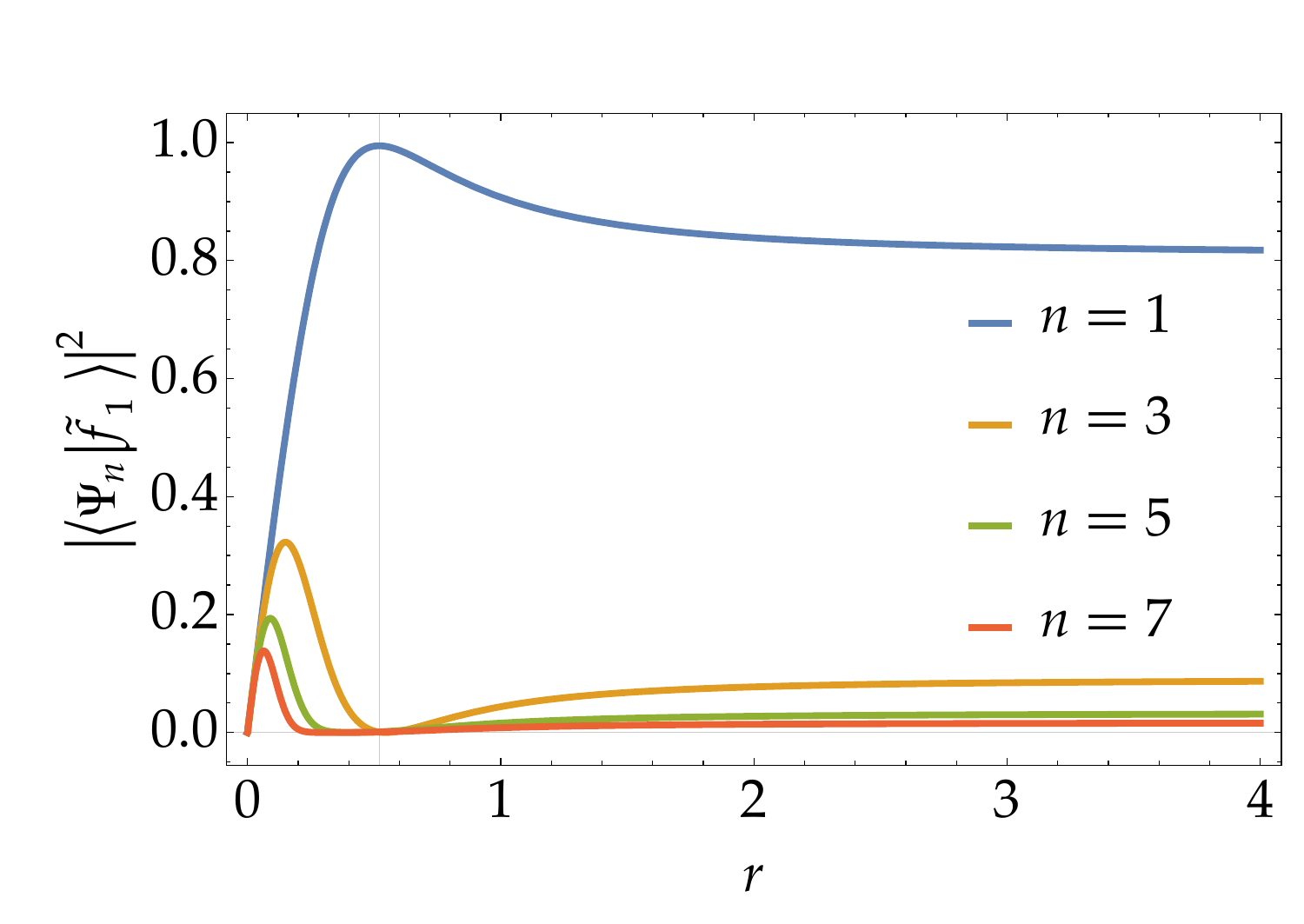}
		\end{minipage}
		\begin{minipage}{0.5\hsize}
			\centering
			\includegraphics[width=0.95\hsize]{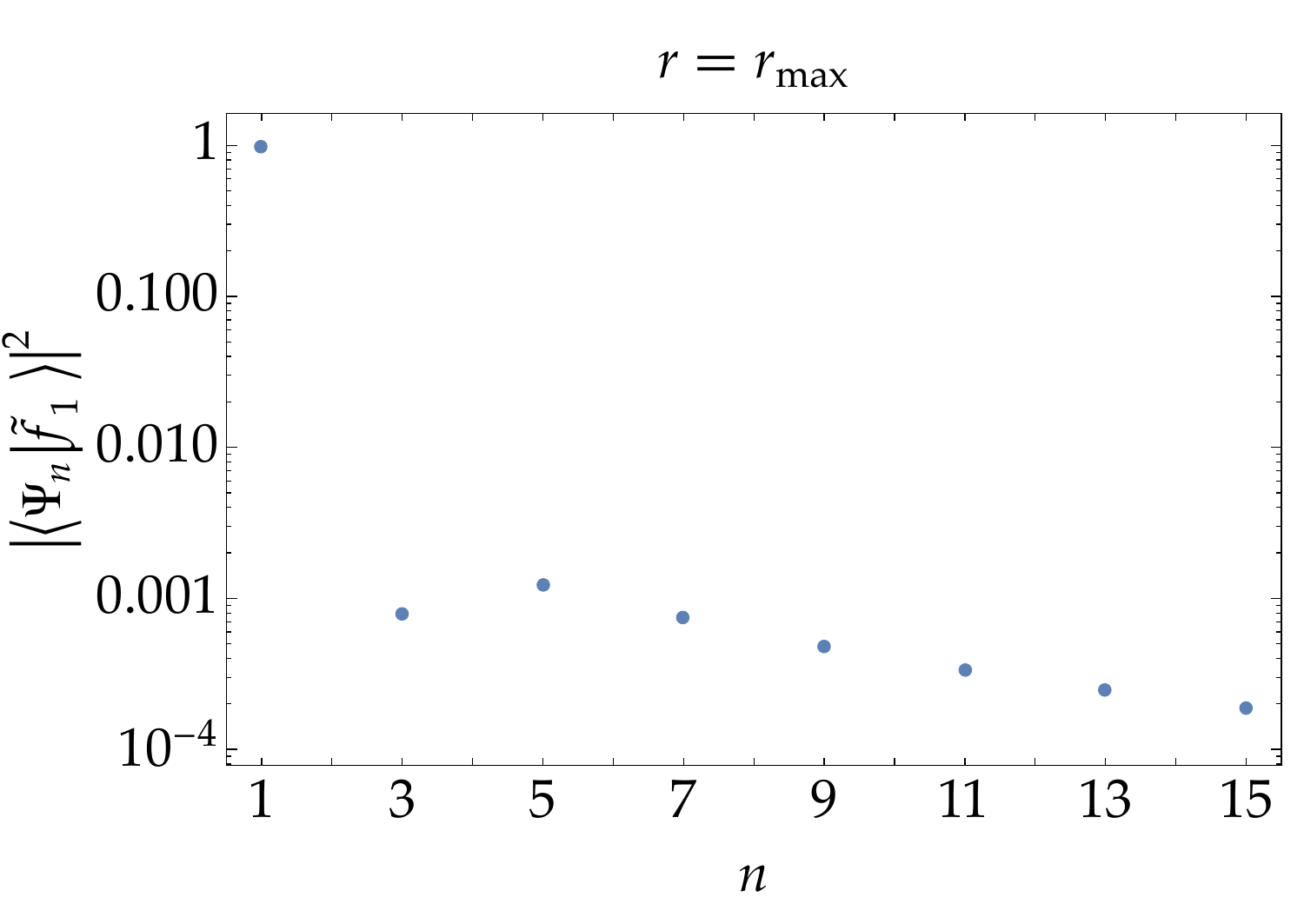}
		\end{minipage}
	\end{tabular}
	\caption{\emph{Left}: the absolute values squared of the inner product~\eqref{eq: f1tPsin QWell hill} as functions of the $r$ parameter of the Gaussian test function $\tilde{f}_1$.
    The inner product of $n=1$ is maximised at $r=r_\umax\sim0.52$ shown by the vertical thin line and reaches $\sim0.99$. \emph{Right}: the inner product for each $n$ at $r=r_\umax$.}
	\label{fig: QWell hilltop}
\end{figure*}

Let us first see the simplest single-field quantum well model as a toy example, reviewing Refs.~\cite{Pattison:2017mbe,Ezquiaga:2019ftu}.
As a zeroth order approximation of hilltop inflation~\cite{Linde:1981mu,Albrecht:1982wi}, we suppose that the inflaton potential has constant and compact support and is bounded by the absorbing (i.e., Dirichlet) boundaries:
\bae{
    v(\phi\in\Omega)=v_0 \qc
    \Omega=[-\phi_\uf,\phi_\uf] \qc
    \partial\Omega=\Bqty{-\phi_\uf,\phi_\uf},
}
where $v_0$ and $\phi_\uf$ are positive parameters.
In this simplest setup, the normalised (i.e., $\braket{\Psi_n|\Psi_m}=\delta_{nm}$) eigenfunctions are easily obtained as
\bae{
    \Psi_n(\phi)=\frac{1}{\sqrt{\phi_\uf}}\sin\bqty{n\pi\frac{\phi+\phi_\uf}{2\phi_\uf}} \qc
    n=1,2,3,\cdots,
}
with the eigenvalues
\bae{
    \Lambda_n=\frac{n^2\pi^2\Mpl^2v_0}{4\phi_\uf^2}.
}
If one chooses the normalised Gaussian
\bae{
    \tilde{f_1}(\phi)=\pqty{\sqrt{\pi}r\phi_\uf\erf\qty(\frac{1}{r})}^{-1/2}\exp(-\frac{\phi^2}{2r^2\phi_\uf^2})
}
with a parameter $r>0$ as the test function,
its inner products with the eigenfunctions, which lead to the $\gamma$ parameter in Eq.~\eqref{eq:f1f1til}, are given by
\bme{\label{eq: f1tPsin QWell hill}
	\braket{\tilde{f}_1|\Psi_n}\simeq\frac{\pi^{1/4}\sqrt{r}}{\sqrt{2\erf(1/r)}}\ee^{-\frac{\pi}{8}(4i(n-1)+n^2\pi r^2)} \\
	\times\pqty{\erf\qty(\frac{2-in\pi r^2}{2\sqrt{2}r})+\erf\qty(\frac{2+in\pi r^2}{2\sqrt{2}r})},
}
for odd $n$ or otherwise zero. 
$\erf(x)=(2/\sqrt{\pi})\int_0^x\ee^{-t^2}\dd{t}$ is the error function.
The $r$ dependence of their absolute values squared is shown in the left panel of Fig.~\ref{fig: QWell hilltop}.
The lowest state $n=1$ dominates the higher modes, taking its maximal value $\sim0.99$ at $r=r_\umax\sim0.52$ shown by the vertical thin line.
The $n$ dependence at $r=r_\umax$ up to a higher value of $n$ is exhibited in the right panel.
One sees that the Gaussian test function can pick up the first eigenmode with the highest probability.

\begin{figure*}
    \centering
    \begin{tabular}{c}
        \begin{minipage}{0.5\hsize}
            \centering
            \includegraphics[width=0.95\hsize]{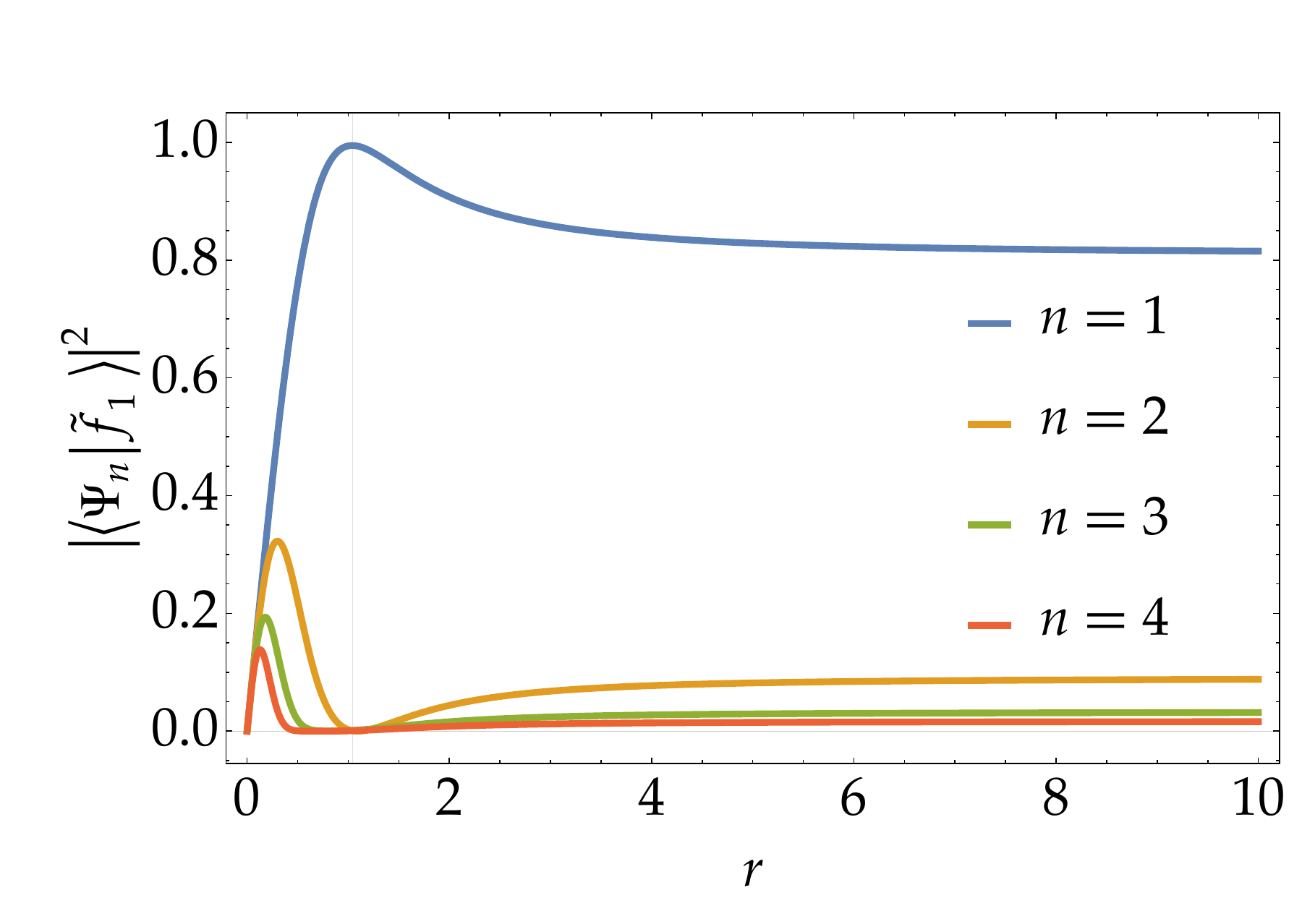}
        \end{minipage}
        \begin{minipage}{0.5\hsize}
            \centering
            \includegraphics[width=0.95\hsize]{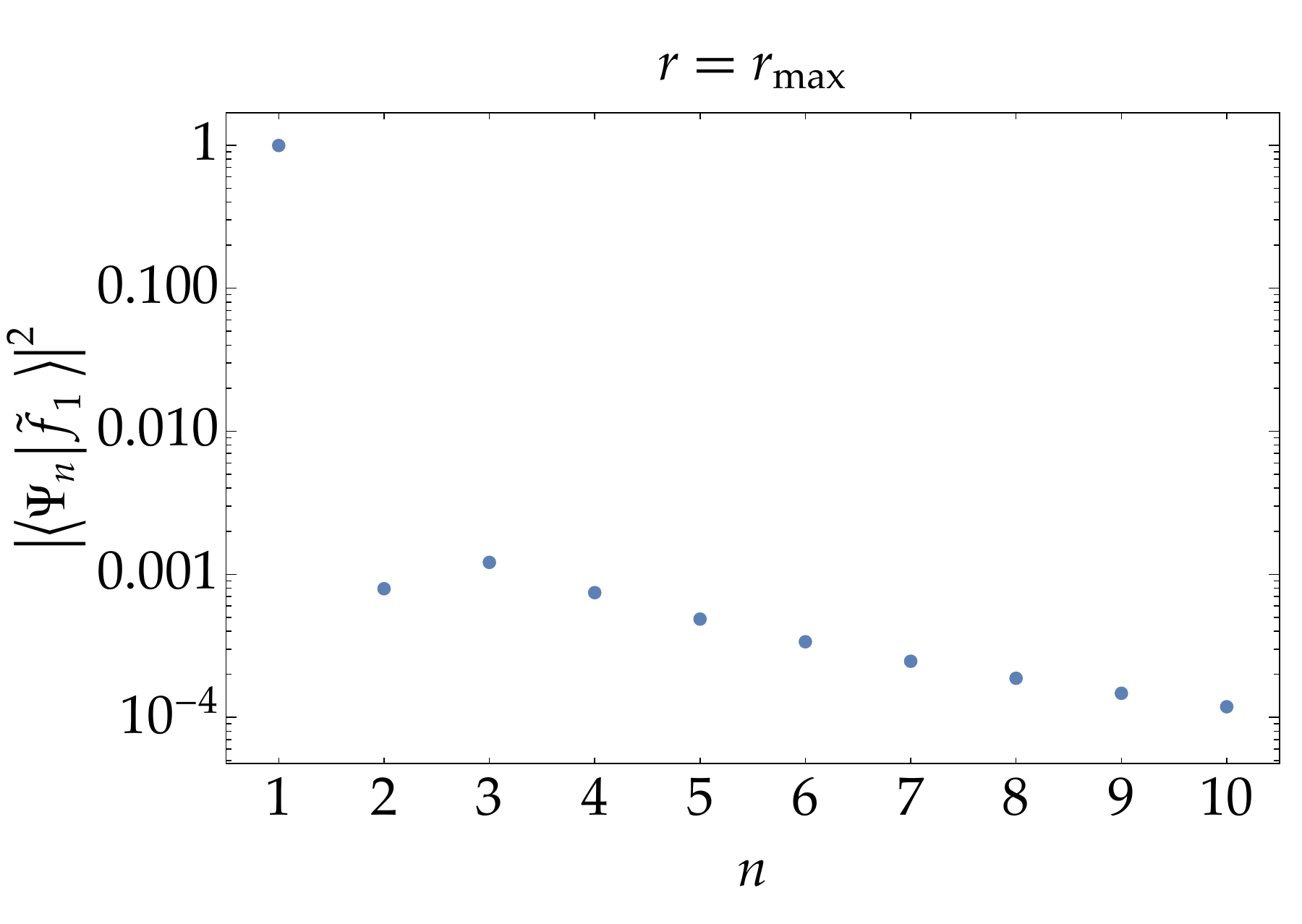}
        \end{minipage}
    \end{tabular}
    \caption{The same plot as Fig.~\ref{fig: QWell hilltop} but for the inner product~\eqref{eq: f1tPsin QWell hill inflec} instead of \eqref{eq: f1tPsin QWell hill}. The vertical thin line in the left panel indicates $r=r_\umax\simeq1.04$ which maximises the inner product for $n=1$.}
    \label{fig: QWell inflection}
\end{figure*}

The flat inflection model (see, e.g., Ref.~\cite{Starobinsky:1992ts,Garcia-Bellido:2017mdw,Ezquiaga:2017fvi,Motohashi:2017kbs}) can also be simulated in the quantum well model just by changing the one boundary condition (we choose $\phi=\phi_\uf$ without loss of generality) from the absorbing to the reflective (i.e., Neumann) one.
The eigenfunctions are then given by
\bae{\label{eq: eigenfunction inflection}
    \Psi_n(\phi)=\frac{1}{\sqrt{\phi_\uf}}\sin\qty[\pqty{n-\frac{1}{2}}\pi\frac{\phi+\phi_\uf}{2\phi_\uf}] \qc
    n=1,2,3,\cdots,
}
with the eigenvalues
\bae{
    \Lambda_n=\frac{(n-1/2)^2\pi^2\Mpl^2v_0}{4\phi_\uf^2}.
}
Its inner product with the Gaussian test function \bae{
    \tilde{f}_1(\phi)\simeq\qty(\frac{\sqrt{\pi}}{2}r\phi_\uf\erf\qty(\frac{2}{r}))^{-1}\exp(-\frac{(\phi-\phi_\uf)^2}{2r^2\phi_\uf^2})
}
reads
\bme{\label{eq: f1tPsin QWell hill inflec}
    \braket{\tilde{f}_1|\Psi_n}=-\frac{(-1)^n\pi^{1/4}\sqrt{r}}{2\sqrt{\erf(2/r)}}\ee^{-\frac{(2n-1)^2\pi^2r^2}{32}} \\
    \times\pqty{\erf\qty(\frac{8+(2n-1)i\pi r^2}{4\sqrt{2}r})+\erf\pqty{\frac{8-(2n-1)i\pi r^2}{4\sqrt{2}r}}}.
}
It is demonstrated in Fig.~\ref{fig: QWell inflection}.
The Gaussian test function can again pick up the lowest mode well.

Note that the reflective boundary is not compatible with our quantum algorithm. 
However, the eigenfunctions~\eqref{eq: eigenfunction inflection} are practically well reproduced by extending the calculation region slightly outside the quantum well, $\phi>\phi_\uf$, with a steep ascent potential and imposing the absorbing boundary condition because $\omega_n^2$ in Eq.~\eqref{eq: omega2} becomes negative for a steep potential and the eigenfunction rapidly damps outside the quantum well.

\subsection{Hybrid inflation}

Let us also see a two-field generalisation called \emph{hybrid inflation}~\cite{Linde:1993cn}.
There, the $\phi$ field rolls down along its potential $V_\phi(\phi)$ with the support potential of the so-called \emph{waterfall} field $\psi$ as
\bae{\label{eq: hybrid V}
    V(\phi,\psi)=V_\phi(\phi)+V_0\bqty{\pqty{1-\pqty{\frac{\psi}{M}}^2}+2\pqty{\frac{\phi\psi}{\phi_\uc M}}^2}.
}
where $V_0$, $M$, and $\phi_\uc$ are the model parameters. Only positive values of $\phi$ are often considered and $V_\phi'(\phi)$ is chosen to be positive so that $\phi$ rolls from a larger value to a smaller value.
$\psi$ is stabilised to $\psi=0$ during $\phi>\phi_\uc$ due to the coupling with $\phi$, while $\psi=0$ is destabilised when $\phi<\phi_\uc$ and $\psi$ rolls down to either potential minima $\psi=\pm M$.
The $\psi$'s fluctuations around $\phi=\phi_\uc$ determine which minima $\psi$ falls to. The whole dynamics is hence stochastic, which makes the model non-trivial.

In the paper, we choose the inflaton potential to have a flat inflection:
\bae{\label{eq: cubic Vphi}
    V_\phi(\phi)=\frac{V_0\beta}{\Mpl^3}(\phi-\phi_\uc)^3,
}
with a parameter $\beta$, which not in hybrid inflation but as a single-field model, have been considered in a previous study \cite{Ezquiaga:2019ftu}.
Then, we numerically find the eigenvalues and eigenfunctions of $\widetilde{\calL_\FP}$, adopting the following parameter values: 
\bege{
    V_0=10^{-15}\Mpl^4 \qc 
    M=10^{16}\,\si{GeV}, \\ \phi_\uc=\sqrt{2}M \qc
    \beta=10^4.
    \label{eq:ParamVal}
}
We take the rectangular absorbing boundary.\footnote{Although the reflective condition is usually set on the high-potential side, whether we set the absorbing or reflective condition has almost no effect on the low eigenvalues and eigenfunctions, because of the steep potential away from the inflection point, as mentioned in Sec.~\ref{sec:qwell}.}
$[\phi_{\rm min}, \phi_{\rm max}]$, the range of $\phi$, is set so that $|\phi-\phi_\uc|/\Mpl \le \beta^{-1/3}$, following Ref.~\cite{Ezquiaga:2019ftu}. 
$\psi_{\rm min}$ and $\psi_{\rm max}$, the endpoints of $\psi$, are set so that at the points, the second slow-roll parameter $\Mpl^2|\partial_\psi^2 v/v|$ for $\psi$ is 1, which means that the inflation ends, with $\phi=\phi_\uc$.
It is important to sufficiently resolve the region near the inflection point, where the stochasticity is significant.
Thus, in $\phi$, we set 1000 equally spaced grid points in $[\phi_{\sto,-},\phi_{\sto,+}]$, and 500 points each in $[\phi_{\rm min},\phi_{\sto,-}]$ and $[\phi_{\sto,+},\phi_{\rm max}]$ so that $\log|\phi-\phi_\uc|$ is equally spaced.
Here, $\phi_{\sto,+}$ is the value of $\phi>\phi_\uc$ at which the stochasticity parameter $\eta_{\sto,\phi}$ for $\phi$ becomes 0.1 with $\psi=0$ and $\phi_{\sto,-}=\phi_\uc-(\phi_{\sto,+}-\phi_\uc)$.
Similarly, in $\psi$, we set 1000 equally spaced grid points in $[\psi_{\sto,-},\psi_{\sto,+}]$, and 500 points each in $[\psi_{\rm min},\psi_{\sto,-}]$ and $[\psi_{\sto,+},\psi_{\rm max}]$ so that $\log|\psi|$ is equally spaced, where $\psi_{\sto,+}>0$ is set so that $\eta_{\sto,\psi}=0.1$ at $(\phi,\psi)=(\phi_{\sto,+},\psi_{\sto,+})$ and $\psi_{\sto,-}=-\psi_{\sto,+}$.
Although such grid points that are not equally spaced in linear scale do not match the setting in our quantum algorithm, it does not matter for our current objective to see the overlap between the test function and the eigenfunctions: it is calculated as
\bae{
    \frac{\left|\sum_{k,l}\tilde{f}_1(\phi^{\rm gr}_k,\psi^{\rm gr}_l)\Psi_n(\phi^{\rm gr}_k,\psi^{\rm gr}_l)\Delta \phi_k \Delta \psi_l\right|^2}{\displaystyle\sum_{k,l}\left|\tilde{f}_1(\phi^{\rm gr}_k,\psi^{\rm gr}_l)\right|^2\Delta \phi_k \Delta \psi_l \times \sum_{k,l}\left|\Psi_n(\phi^{\rm gr}_k,\psi^{\rm gr}_l)\right|^2\Delta \phi_k \Delta \psi_l},
}
where $\phi^{\rm gr}_k$ (resp. $\psi^{\rm gr}_k$) is the $k$-th grid point in $\phi$ (resp. $\psi$) and $\Delta \phi_k=\phi^{\rm gr}_{k+1}-\phi^{\rm gr}_k$ and $\Delta \psi_k=\psi^{\rm gr}_{k+1}-\psi^{\rm gr}_k$.
We will consider extending our quantum algorithm so that nonequidistant grid points can be dealt with in future works.

By using \texttt{eigs} in SciPy~\cite{2020SciPy-NMeth}, we find the first ten eigenvalues and eigenfunctions, which are shown in Fig.~\ref{fig:eigenval} and \ref{fig:eigenfunc} respectively.
Noting that the eigenfunctions have peaks in the high-stochasticity region, we take the following Gaussian test function:
\bae{\label{eq:testFuncHyb}
    \tilde{f}_1(\phi,\psi)\propto\exp\left(-\frac{(\phi-\phi_\uc)^2}{2(\phi_{\sto,+}-\phi_\uc)^2}-\frac{\psi^2}{2\psi_{\sto,+}^2}\right)
}
The overlaps between this and the eigenfunctions are shown in Fig.~\ref{fig:overlapHyb}.
One sees that the overlap is significant for the first eigenfunction, which implies that with this test function, our quantum algorithm will also work well in this model.

\begin{figure}
	\centering
        \includegraphics[width=1\hsize]{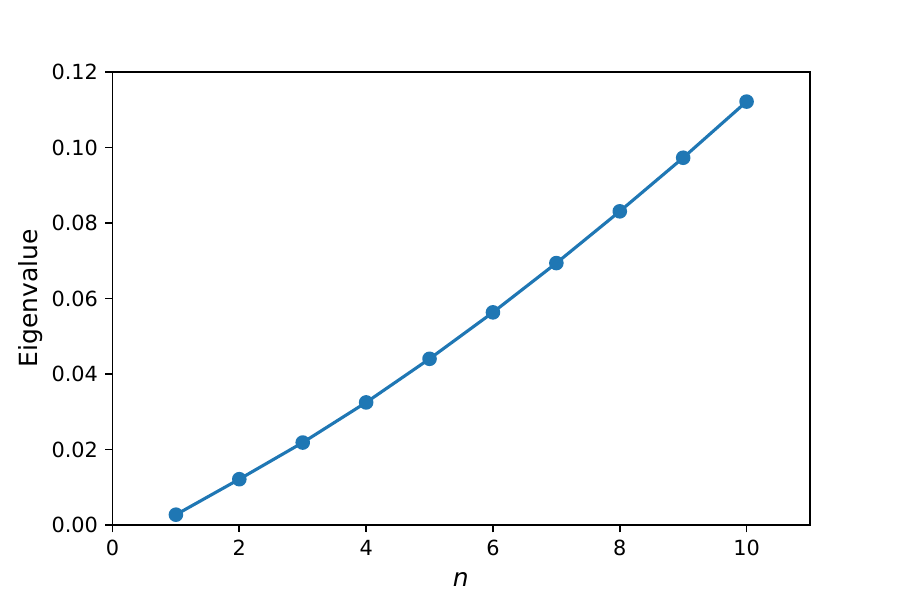}
	\caption{The first %to tenth 
    ten eigenvalues of $\widetilde{\calL_\FP}$ with $V$ in Eqs.~\eqref{eq: hybrid V} and \eqref{eq: cubic Vphi} for the parameters in Eq.~\eqref{eq:ParamVal}.}
	\label{fig:eigenval}
\end{figure}

\begin{figure*}
	\centering
	\subfigure[$n=1$]{
        \includegraphics[width=1\columnwidth]{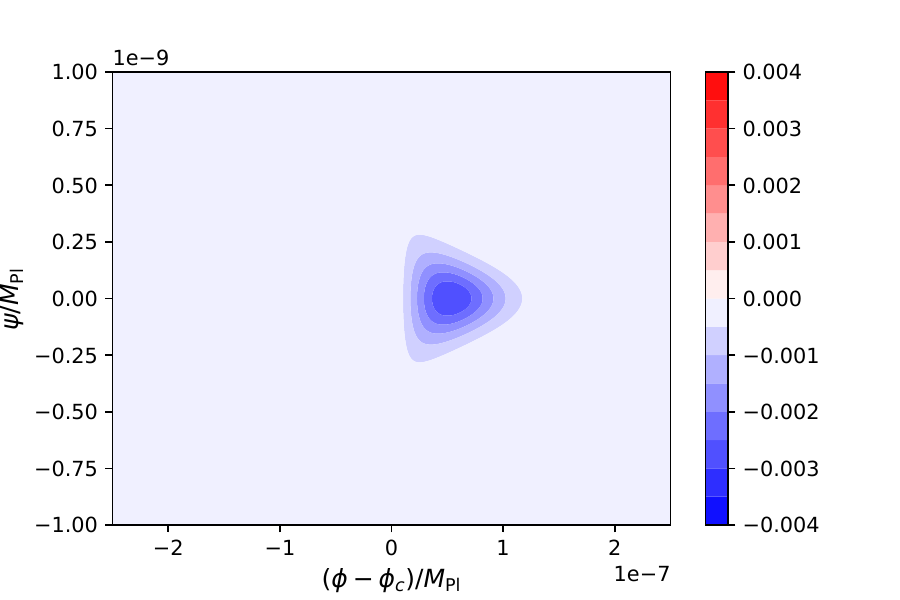}
    }
    \subfigure[$n=2$]{
        \includegraphics[width=1\columnwidth]{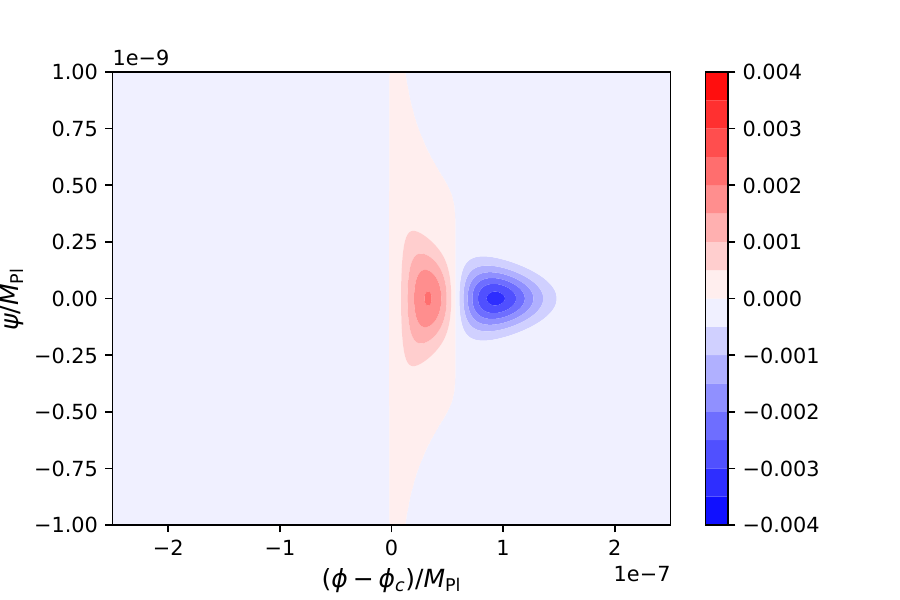}
	}
    \subfigure[$n=3$]{
			\includegraphics[width=1\columnwidth]{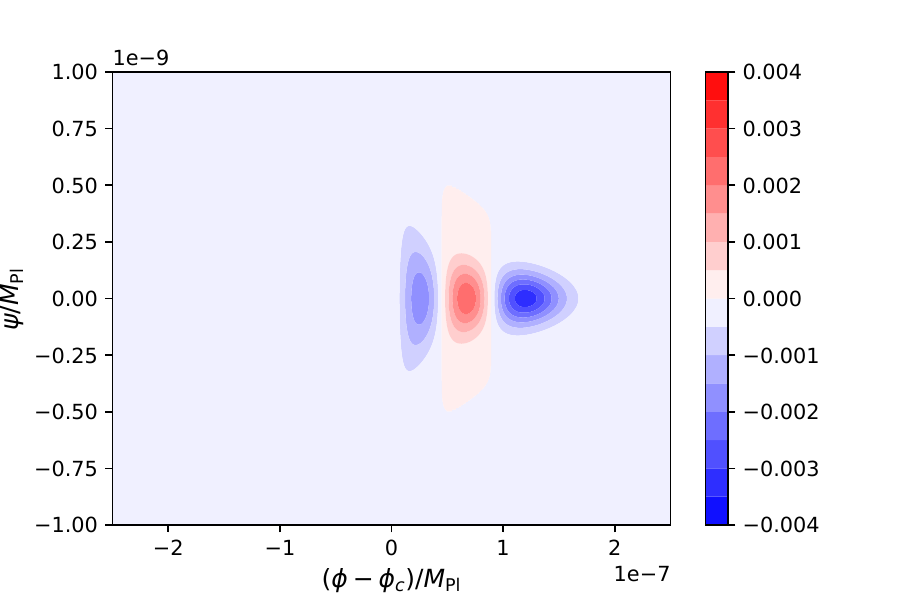}
	}
	\subfigure[$n=4$]{
			\includegraphics[width=1\columnwidth]{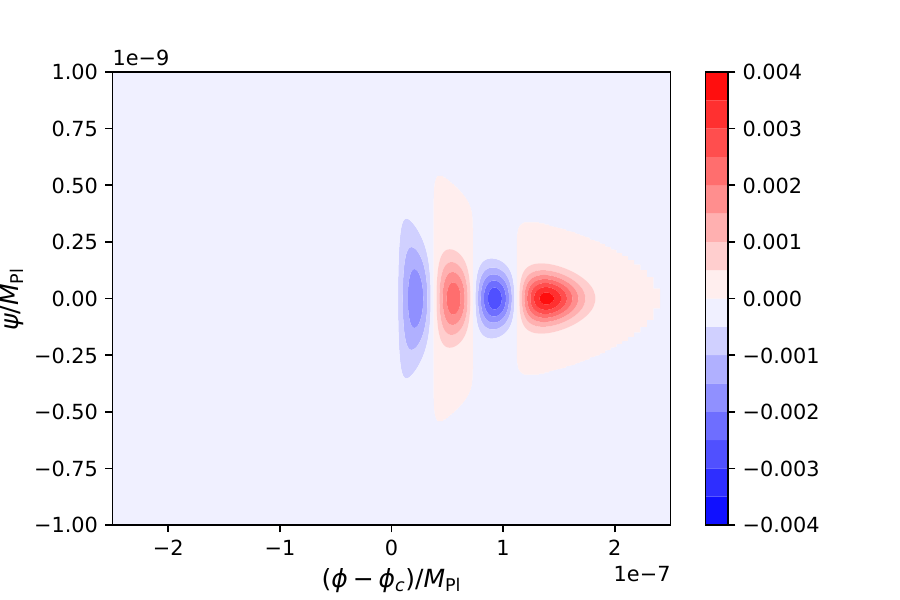}
	}
	\caption{The %$n$th 
    first four eigenfunctions of the same $\widetilde{\calL_\FP}$ as Fig.~\ref{fig:eigenval}.}
	\label{fig:eigenfunc}
\end{figure*}

\begin{figure}
	\centering
        \includegraphics[width=1\hsize]{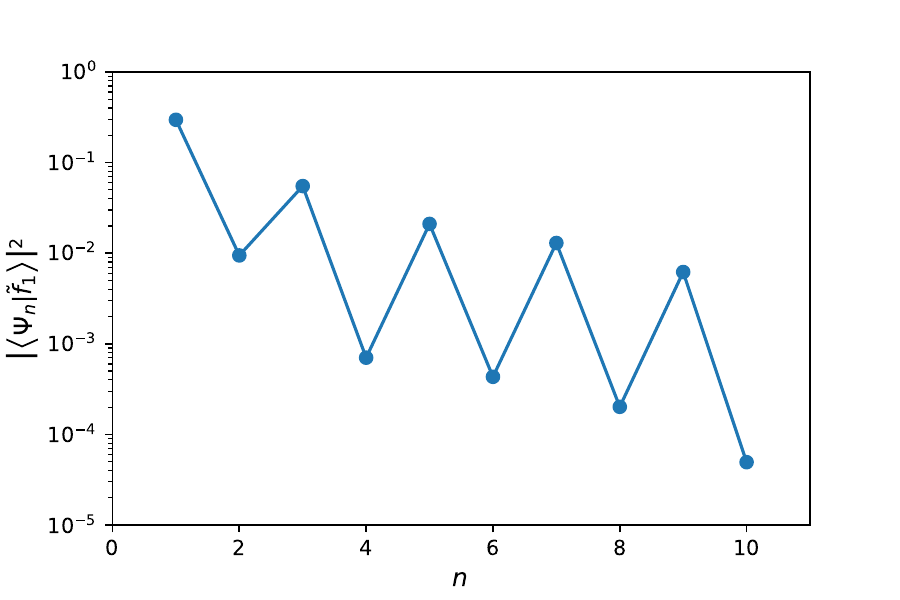}
	\caption{The squared inner products between the quantum states encoding the Gaussian test function $\tilde{f}_1$ in Eq. \eqref{eq:testFuncHyb} and those encoding $n$th eigenfunctions of the same $\widetilde{\calL_\FP}$ as Fig.~\ref{fig:eigenval}.}
	\label{fig:overlapHyb}
\end{figure}

\section{Summary \label{sec:sum}}

In this paper, we considered a quantum algorithm for calculating the first eigenvalue of differential operators.
Under the finite difference approximation of a given operator, this problem boils down to the matrix eigenvalue problem, but in multi-dimensional cases, it is computationally demanding because of the exponential increase of the size of the approximating matrix.
Then, we proposed a quantum algorithm for this task, leveraging the \ac{QSVT}-based quantum algorithm for finding the first eigenvalue of matrices in Ref.~\cite{Lin2020nearoptimalground}.
Our quantum algorithm has the query complexity scaling as $\widetilde{O}(1/\epsilon^2)$ on the accuracy $\epsilon$ in the eigenvalue, which shows the improvement compared to the existing quantum algorithm for the same task with complexity of order $\widetilde{O}(1/\epsilon^3)$.

As a potential application target for our algorithm, we considered a problem in cosmology, finding the eigenvalue of the adjoint \ac{FP} operator $\calL_\FP^\dagger$ in stochastic inflation, which is related to the tail shape of the \ac{PDF} of the primordial perturbation.
Although we cannot run our algorithm for concrete instances of this problem now because of the absence of large-scale fault-tolerant quantum computers, we conducted numerical demonstrations for some problem instances in shochastic inflation, to show that our method is promising.
We take Gaussian functions as test functions, expecting the first eigenfunctions to have a single-bump shape, and saw that they in fact overlap with the first eigenfunctions well in the considered problems, which implies that one of the conditions for our method to work would be satisfied.

Nevertheless, refining the way to set the test function should be considered in future works, since it is not obvious that the current approach will also work in more complicated problems, such as higher-dimensional ones.
Preparing the test function is an issue also in ground energy estimation in condensed matter physics and chemistry, where various approaches have been proposed: e.g., adiabatic state preparation considered in \cite{yoshioka2024hunting}.
It will be interesting and significant to explore applications of such advanced state preparation methods to differential operator eigenproblems, learning from existing studies in other fields.

\section*{Acknowledgements}

KM is supported by MEXT Quantum Leap Flagship Program (MEXT Q-LEAP) Grant no. JPMXS0120319794, JSPS KAKENHI Grant no. JP22K11924, and JST COI-NEXT Program Grant No. JPMJPF2014.
YT is supported by JSPS KAKENHI Grant
No.~JP24K07047.

\appendix

\section{Proofs}

\subsection{Proof of Theorem \ref{th:EigConv} \label{sec:ProofEigConv}}

Before the proof, we introduce some kinds of norms and smoothness classes of functions.
Let $\Omega\subset\mathbb{R}^d$ be a bounded region.
For a function $g:\Omega\rightarrow\mathbb{R}$, we define 
\begin{equation}
|g|_{0,\Omega}%:=
\coloneqq\sup_{\mathbf{x}\in\Omega}|g(\mathbf{x})|    
\end{equation}
and
\begin{align}
    |g|_{a,\Omega}&%:=
    \coloneqq\sum_{|\boldsymbol{\alpha}| \le \lceil a \rceil -1}  |D^{\boldsymbol{\alpha}} g(\mathbf{x})|_{0,\Omega} \nonumber \\
    &+\sum_{|\boldsymbol{\alpha}|=\lceil a \rceil -1} \sup_{\mathbf{x},\mathbf{y}\in\Omega} \frac{|D^{\boldsymbol{\alpha}}g(\mathbf{x})-D^{\boldsymbol{\alpha}}g(\mathbf{y})|}{\|\mathbf{x}-\mathbf{y}\|^{a-\lceil a \rceil +1}}
\end{align}
with $a\in\mathbb{R}_+$.
We say that $g \in C^a\left(\overline{\Omega}\right)$ if $|g|_{a,\Omega}<\infty$, and that $g \in C^a\left(\Omega\right)$ if $g \in C^a\left(\overline{O}\right)$ for every $\overline{O}\subset\Omega$.
Obviously, any function in $C^a\left(\overline{\Omega}\right)$ needs to be $(\lceil a \rceil-1)$-times differentiable.
We can easily see that any $\lceil a \rceil$-times continuously differentiable function on $\overline{\Omega}$ is in $C^a\left(\overline{\Omega}\right) \cap C^a\left(\Omega\right)$.

Then, the proof of Theorem \ref{th:EigConv} is as follows.

\begin{proof}[Proof of Theorem \ref{th:EigConv}]
    Theorem \ref{th:EigConv} is obtained by just applying Theorem 5.1 in \cite{Kuttler1970} to Problem \ref{prob:eigen} with the finite difference approximation in Eq. \eqref{eq:FDMat}.
    Now, we just check that the conditions for applying Theorem 6.4 in \cite{Kuttler1970} are satisfied.

    \begin{itemize}
        \item $L_{n_\mathrm{gr}}$ acts locally, that is, any entry $(L_{n_\mathrm{gr}})_{\mathbf{j}_1,\mathbf{j}_2}$ corresponding to the grid points $\mathbf{x}^\mathrm{gr}_{\mathbf{j}_1}$ and $\mathbf{x}^\mathrm{gr}_{\mathbf{j}_2}$ such that $\|\mathbf{x}^\mathrm{gr}_{\mathbf{j}_1}-\mathbf{x}^\mathrm{gr}_{\mathbf{j}_2}\| > h$ is 0.

        \item $L_{n_\mathrm{gr}}$ is symmetric.
    
        \item Condition (i) (see \cite{Kuttler1970} for details; the same applies below)

        In the grid point set $\mathcal{D}_{n_\mathrm{gr}}$ given in Sec. \ref{sec:FDAppEigen}, for any adjacent point pair $(\mathbf{x}^\mathrm{gr}_{\mathbf{j}_1}, \mathbf{x}^\mathrm{gr}_{\mathbf{j}_2})$, where $\mathbf{j}_1=\mathbf{j}_2 \pm \mathbf{e}_i$ for some $i\in[d]$, the corresponding entry $(L_{n_\mathrm{gr}})_{\mathbf{j}_1,\mathbf{j}_2}$ in $L_{n_\mathrm{gr}}$ is nonzero.
        Besides, for any $\mathbf{x}^\mathrm{gr}_{\mathbf{j}}, \mathbf{x}^\mathrm{gr}_{\mathbf{j}^\prime}\in\mathcal{D}_{n_\mathrm{gr}}$, we can move from $\mathbf{x}^\mathrm{gr}_{\mathbf{j}}$ to $\mathbf{x}^\mathrm{gr}_{\mathbf{j}^\prime}$ by some sequence of transitions to an adjacent point.
        Combining these observations, we see that the condition (i) is satisfied.

        \item Condition (ii)

        We have
        \begin{align}
            &\sum_{\substack{\mathbf{j}^\prime\in\mathcal{D}_{n_\mathrm{gr}} \\ \mathbf{j}^\prime \ne \mathbf{j}}} \left|h^2 (L_{n_\mathrm{gr}})_{\mathbf{j},\mathbf{j}^\prime}\right| \le 2d a_{\mathrm{max}}, \nonumber \\
            &\frac{1}{h^2 \left| (L_{n_\mathrm{gr}})_{\mathbf{j},\mathbf{j}} \right|} \le \frac{1}{2d a_{\mathrm{min}}}, \nonumber \\
            &\sum_{\substack{\mathbf{j}^\prime\in\mathcal{D}_{n_\mathrm{gr}} \\ \mathbf{j}^\prime \ne \mathbf{j}}} \left|\frac{ (L_{n_\mathrm{gr}})_{\mathbf{j},\mathbf{j}^\prime}}{(L_{n_\mathrm{gr}})_{\mathbf{j},\mathbf{j}}}\right| \le 1
        \end{align}
        for any $\mathbf{j}\in[n_\mathrm{gr}]_0^d$, where
        $a_{\mathrm{min}}%:=
        \coloneqq\min_{\substack{i\in[d] \\ \mathbf{x}\in\overline{\mathcal{D}}}} a_i(\mathbf{x})$.
        This means that the condition (ii) is satisfied\footnote{Note that, in our definition, $L_{n_\mathrm{gr}}$ has no entry corresponding to a point in $\partial \mathcal{D}$. Even if we make it have such entries as in \cite{Kuttler1970}, they are zero because of the Dirichlet boundary condition and thus have nothing to do with the condition (ii).}.

        \item Condition (iii)

        It is obvious that $\sum_{\substack{\mathbf{j}^\prime\in\mathcal{D}_{n_\mathrm{gr}} \\ \mathbf{j}^\prime \ne \mathbf{j}}} \left| (L_{n_\mathrm{gr}})_{\mathbf{j},\mathbf{j}^\prime} \right| \le (L_{n_\mathrm{gr}})_{\mathbf{j},\mathbf{j}}$ for any $\mathbf{j}%:=
        \coloneqq[n_\mathrm{gr}]_0^d$, which means that the first part of this condition is satisfied.
        The second part is irrelevant in the current case, since now $\mathcal{D}_{n_\mathrm{gr}}^*%:=
        \coloneqq\mathcal{D}_{n_\mathrm{gr}}-\mathcal{D}_{n_\mathrm{gr}}^\prime$ is empty\footnote{Note the differences of the notations in this paper and \cite{Kuttler1970}. $\mathcal{D}_{n_\mathrm{gr}}^\prime$ and $\mathcal{D}_{n_\mathrm{gr}}^*$ correspond to $\Omega_h^\prime$ and $\Omega_h^*$ in \cite{Kuttler1970}, respectively.}.
        Here, $\mathcal{D}_{n_\mathrm{gr}}^\prime%:=
        \coloneqq\{\mathbf{x} \in \mathcal{D}_{n_\mathrm{gr}} \ | \ \inf_{\mathbf{x}^\prime \in \partial\mathcal{D}} \|\mathbf{x}-\mathbf{x}^\prime\| \ge h \}$, which is now equal to $\mathcal{D}_{n_\mathrm{gr}}$.

        \item Condition (iv)

        This is also irrelevant in the current case since $\mathcal{D}_{n_\mathrm{gr}}^*=\emptyset$.

        \item Condition (v)

        It is obvious that $(L_{n_\mathrm{gr}})_{\mathbf{j},\mathbf{j}^\prime}\le0$ for any $\mathbf{j},\mathbf{j}^\prime\in[n_\mathrm{gr}]_0^d$ such that $\mathbf{j}\ne\mathbf{j}^\prime$, which means this condition is satisfied.

        \item Let us check that, in the words of \cite{Kuttler1970}, $L_{n_\mathrm{gr}}$ is consistent of order 2 with $\mathcal{L}$ in $\mathcal{D}_{n_\mathrm{gr}}$, that is, there exists a constant $C$ independent of $h$ such that
        \begin{equation}
            \left|\mathcal{L}f\left(\mathbf{x}^\mathrm{gr}_{\mathbf{j}}\right)-\left(L_{n_\mathrm{gr}}\mathbf{v}_{f,n_\mathrm{gr}}\right)_{\mathbf{j}}\right|\le C h^\mu |f|_{\mu+2,S_{h,\mathbf{j}}\cap\mathcal{D}}
            \label{eq:consis}
        \end{equation}
        holds for any $\mu\in(0,2]$,  $\mathbf{j}\in[n_\mathrm{gr}]_0^d$ and $f \in C^{\mu+2}\left(\overline{S_{h,\mathbf{j}}\cap\mathcal{D}}\right)$.
        Here, $S_{h,\mathbf{j}}%:=
        \coloneqq\{\mathbf{x}\in\mathbb{R}^d \ | \ \|\mathbf{x}-\mathbf{x}^\mathrm{gr}_{\mathbf{j}}\|<h\}$ is the sphere of radius $h$ centered at $\mathbf{x}^\mathrm{gr}_{\mathbf{j}}$.

        $f \in C^{\mu+2}\left(\overline{S_{h,\mathbf{j}}\cap\mathcal{D}}\right)$ implies the following.
        First, $f$ is $\tilde{\mu}$-times differentiable, where $\tilde{\mu}=\lceil \mu+2 \rceil-1$, and thus, as we can see using Taylor's theorem,
        \begin{equation}
            f(\mathbf{x}^\mathrm{gr}_\mathbf{j}+\delta\mathbf{e}_i) = \sum_{n=0}^{\tilde{\mu}-1} \frac{1}{n!}\frac{\partial^n f}{\partial x_i^n}(\mathbf{x}^\mathrm{gr}_\mathbf{j}) \delta^n+ \frac{1}{\tilde{\mu}!}\frac{\partial^{\tilde{\mu}} f}{\partial x_i^{\tilde{\mu}}}(\mathbf{x}^\mathrm{gr}_\mathbf{j}+\delta^\prime\mathbf{e}_i) \delta^{\tilde{\mu}}
            \label{eq:Taylor}
        \end{equation}
        holds for any $\mathbf{j}\in[n_\mathrm{gr}]_0^d$, $i\in[d]$, and $\delta\in[-h,h]$, with some real number $\delta^\prime$ that is between 0 and $\delta$ and dependent on $i,\mathbf{x}^\mathrm{gr}_\mathbf{j}$ and $\delta$.
        Second, for any $\mathbf{x},\mathbf{y}\in S_{h,\mathbf{j}}\cap\mathcal{D}$ and $i\in[d]$,
        \begin{equation}
            \left|\frac{\partial^{\tilde{\mu}} f}{\partial x_i^{\tilde{\mu}}}(\mathbf{x})-\frac{\partial^{\tilde{\mu}} f}{\partial x_i^{\tilde{\mu}}}(\mathbf{y})\right| \le |f|_{\mu+2,S_{h,\mathbf{j}}\cap\mathcal{D}} \|\mathbf{x}-\mathbf{y}\|^{\mu+2 - \lceil \mu+2 \rceil +1}
            \label{eq:fprprprBnd}
        \end{equation}
        holds.
        Combining Eqs. \eqref{eq:Taylor} and \eqref{eq:fprprprBnd}, we have
        \begin{equation}
            f(\mathbf{x}^\mathrm{gr}_\mathbf{j}+\delta\mathbf{e}_i) = \sum_{n=0}^{\tilde{\mu}} \frac{1}{n!}\frac{\partial^n f}{\partial x_i^n}(\mathbf{x}^\mathrm{gr}_\mathbf{j}) \delta^n +R,
            \label{eq:TaylorMod}
        \end{equation}
        with the residual term $R$ bounded as
        \begin{equation}
            |R|\le \frac{1}{\tilde{\mu}!} |f|_{\mu+2,S_{h,\mathbf{j}}\cap\mathcal{D}} |\delta|^{\mu+2}.
        \end{equation} 
        For $a_i$, since it is four-times continuously differentiable on $\mathcal{D}$ and $\tilde{\mu}\le3$, Taylor's theorem implies that, for any $n\in[\tilde{\mu}]$,
        \begin{align}
        a_i(\mathbf{x}^\mathrm{gr}_\mathbf{j}+\delta\mathbf{e}_i) &= \sum_{m=0}^{n}\frac{1}{m!}\frac{\partial^m a_i}{\partial x_i^m}(\mathbf{x}^\mathrm{gr}_\mathbf{j}) \delta^m \nonumber \\
        & +\frac{1}{(n+1)!}\frac{\partial^{n+1} a_i}{\partial x_i^{n+1}}(\mathbf{x}^\mathrm{gr}_\mathbf{j}  +\delta^{\prime\prime}\mathbf{e}_i) \delta^{n+1}
            \label{eq:Taylora}
        \end{align}
        holds with some real number $\delta^{\prime\prime}$ that is between $0$ and $\delta$ and dependent on $i$, $n$, $\mathbf{x}^\mathrm{gr}_\mathbf{j}$ and $\delta$.
        By using Eqs. \eqref{eq:TaylorMod} and \eqref{eq:Taylora} with $\delta=\pm h$ and $\delta=\pm\frac{h}{2}$ in Eq. \eqref{eq:FDMatExp}, we obtain
        \begin{align}
            & \left(L_{n_\mathrm{gr}}\mathbf{v}_{f,n_\mathrm{gr}}\right)_{\mathbf{j}}= \nonumber \\
            & \ \sum_{i=1}^d\left(\frac{\partial a_i}{\partial x_i}(\mathbf{x}^\mathrm{gr}_\mathbf{j})\frac{\partial f}{\partial x_i}(\mathbf{x}^\mathrm{gr}_\mathbf{j})+a_i(\mathbf{x}^\mathrm{gr}_\mathbf{j})\frac{\partial^2 f}{\partial x_i^2}(\mathbf{x}^\mathrm{gr}_\mathbf{j})\right)   \nonumber \\
            & \ +a_0(\mathbf{x}^\mathrm{gr}_\mathbf{j})f(\mathbf{x}^\mathrm{gr}_\mathbf{j}) + R^\prime
        \end{align}
        Here, the residual term $R^\prime$ is bounded as
        \begin{align}
            & |R^\prime| \le \nonumber \\
            & \ \sum_{i=1}^d\sum_{n=0}^{\tilde{\mu}} \frac{1}{2^{\tilde{\mu}-n}(\tilde{\mu}-n+1)!n!}a^{(\tilde{\mu}-n+1)}_\mathrm{max}\left|\frac{\partial^n f}{\partial x_i^n}\right|_{0,S_{h,\mathbf{j}}\cap\mathcal{D}}h^{\tilde{\mu}-1} \nonumber \\
            & \ + \frac{2d}{\tilde{\mu}!} |f|_{\mu+2,S_{h,\mathbf{j}}\cap\mathcal{D}} h^{\mu},
        \end{align}
        where, for $n\in\mathbb{N}$, $a^{(n)}_\mathrm{max}%:=
        \coloneqq\max_{i\in[d]} \left|\frac{\partial^n a_i}{\partial x_i^n}\right|_{0,\mathcal{D}}$.
        This implies that Eq. \eqref{eq:consis} holds with $C$
        \begin{equation}
            C%:=
            \coloneqq\max_{n\in[\tilde{\mu}+1]_0} \frac{a^{(\tilde{\mu}-n+1)}_\mathrm{max}}{2^{\tilde{\mu}-n}(\tilde{\mu}-n+1)!n!}\times (U-L)^{\lceil \mu+2 \rceil - \mu-2} + \frac{2d}{\tilde{\mu}!}.
        \end{equation}

        \item Lastly, let us check that the each eigenfunction $f_k$ of $\mathcal{L}$ is in $C^4(\overline{\mathcal{D}})$.
        Since $a_0,%...
        \cdots,a_d$ are now four-times continuously differentiable and thus twice continuously differentiable, they are in $C^2(\overline{\mathcal{D}}) \cap C^2(\mathcal{D})$.
        Then, Theorem 6.3 in \cite{Kuttler1970} implies that $f_k$ is in $C^4(\overline{\mathcal{D}}) \cap C^4(\mathcal{D})$.
    \end{itemize}
\end{proof}

\subsection{Proof of Corollary \ref{co:GEESparse} \label{sec:ProofGEESparse}}

\begin{proof}[Proof of Corollary \ref{co:GEESparse}]
    
    Note that, given $(\alpha,a,0)$-block-encoding $U_H$ of $H$, the quantum algorithm in \cite{Lin2020nearoptimalground} relies on a unitary $\mathrm{PROJ}\left(\mu,\frac{\epsilon}{2\alpha},\epsilon^\prime\right)$ constructed with $U_H$, where $\mu\in\mathbb{R}$ is any real number, $\epsilon^\prime$ is set to $\epsilon^\prime=\gamma/2$, and
    \begin{align}
        &\left\|(\bra{0}^{\otimes(a+3)}\otimes I_{2^n})\mathrm{PROJ}\left(\mu,\frac{\epsilon}{2\alpha},\epsilon^\prime\right)\ket{0}^{\otimes(a+3)}\ket{\phi_1}\right\| \nonumber \\
        & \quad 
        \begin{cases}
            \ge \gamma - \frac{\epsilon^\prime}{2} & ; \ \mathrm{if} \ \lambda_1 \le \mu - \epsilon \\
            \le \frac{\epsilon^\prime}{2} & ; \ \mathrm{if} \ \lambda_1 \ge \mu + \epsilon
        \end{cases}
        \label{eq:PROJ}
    \end{align}
    holds.
    $U_H$ is used only in this unitary, and thus, if we can construct a unitary $\widetilde{\mathrm{PROJ}}\left(\mu,\frac{\epsilon}{2\alpha},\epsilon^\prime\right)$ that has the same property as Eq. \eqref{eq:PROJ} using $O^H_\mathrm{row}$, $O^H_\mathrm{col}$ and $O^H_\mathrm{ent}$ instead of $U_H$, we can run the quantum algorithm replacing $\mathrm{PROJ}\left(\mu,\frac{\epsilon}{2\alpha},\epsilon^\prime\right)$ with $\widetilde{\mathrm{PROJ}}\left(\mu,\frac{\epsilon}{2\alpha},\epsilon^\prime\right)$.
    In particular, it suffices that we have $\widetilde{\mathrm{PROJ}}\left(\mu,\frac{\epsilon}{2\alpha},\epsilon^\prime\right)$ such that
    \begin{align}
        &\left\|(\bra{0}^{\otimes(a+3)}\otimes I_{2^n})\widetilde{\mathrm{PROJ}}\left(\mu,\frac{\epsilon}{2\alpha},\epsilon^\prime\right)(\ket{0}^{\otimes(a+3)}\otimes I_{2^n}) - \right. \nonumber \\
        &\quad \left. (\bra{0}^{\otimes(a+3)}\otimes I_{2^n})\mathrm{PROJ}\left(\mu,\frac{\epsilon}{2\alpha},\frac{\epsilon^\prime}{2}\right)(\ket{0}^{\otimes(a+3)}\otimes I_{2^n})\right\| \nonumber \\
        &\le \frac{\epsilon^\prime}{4}, \label{eq:PROJtil}
    \end{align}
    which, as we can see by simple algebra, leads to $\widetilde{\mathrm{PROJ}}\left(\mu,\frac{\epsilon}{2\alpha},\epsilon^\prime\right)$ satisfying the property like Eq. \eqref{eq:PROJ}.    

    Then, let us consider how to construct such $\widetilde{\mathrm{PROJ}}\left(\mu,\frac{\epsilon}{2\alpha},\epsilon^\prime\right)$, fixing $\alpha$ to $\tilde{\alpha}%:=
    \coloneqq s\|H\|_\mathrm{max}$ and $a$ to $\tilde{a}%:=
    \coloneqq n+3$.
    On the one hand, we note that
    \begin{align}
        & \mathrm{PROJ}\left(\mu,\frac{\epsilon}{2\tilde{\alpha}},\frac{\epsilon^\prime}{2}\right)= \nonumber \\
        \quad & (\mathrm{Had} \otimes I_{2^{n+\tilde{a}+3}}) \otimes \nonumber \\
        \quad & \left(\ket{0}\bra{0} \otimes I_{2^{n+\tilde{a}+3}} + \ket{1}\bra{1} \otimes \mathrm{REF}\left(\mu,\frac{\epsilon}{2\tilde{\alpha}},\frac{\epsilon^\prime}{2}\right)\right) \otimes\nonumber \\
        \quad & (\mathrm{Had} \otimes I_{2^{n+\tilde{a}+3}})
        \label{eq:PROJDef}
    \end{align}
    by the definition in \cite{Lin2020nearoptimalground}.
    Here, $\mathrm{Had}$ is a Hadamard gate, and $\mathrm{REF}\left(\mu,\frac{\epsilon}{2\tilde{\alpha}},\frac{\epsilon^\prime}{2}\right)$ is a $(\tilde{\alpha},\tilde{a}+2,0)$-block-encoding of $-S\left(\frac{H-\mu I}{\tilde{\alpha}+|\mu|};\frac{\epsilon}{2\tilde{\alpha}},\frac{\epsilon^\prime}{2}\right)$ with $S\left(\cdot;\frac{\epsilon}{2\tilde{\alpha}},\frac{\epsilon^\prime}{2}\right)$ being some polynomial of degree $d_{\tilde{\alpha},\epsilon,\epsilon^\prime}=O\left(\frac{\tilde{\alpha}}{\epsilon}\log\left(\frac{1}{\epsilon^\prime}\right)\right)$ (see \cite{Lin2020nearoptimalground} for the details).
    On the other hand, because of Theorem \ref{th:BlEncSp}, we can construct a $(\tilde{\alpha},\tilde{a},\tilde{\epsilon})$-block-encoding $\tilde{U}_H$ of $H$ with $O^H_\mathrm{row}$, $O^H_\mathrm{col}$ and $O^H_\mathrm{ent}$, where 
    \begin{equation}
        \tilde{\epsilon}%:=
        \coloneqq\left(\frac{\epsilon^\prime}{8d_{\tilde{\alpha},\epsilon,\epsilon^\prime}}\right)^2\tilde{\alpha}.
    \end{equation}
    $\tilde{U}_H$ can be regarded as a $(\tilde{\alpha},\tilde{a},0)$-block-encoding of some matrix $\tilde{H}\in\mathbb{C}^{2^n \times 2^n}$ such that $\|\tilde{H}-H\|\le\tilde{\epsilon}$.
    Thus, because of Lemma 22 in the full version of \cite{Gilyen2019}, replacing $U_H$ in $\mathrm{REF}\left(\mu,\frac{\epsilon}{2\tilde{\alpha}},\frac{\epsilon^\prime}{2}\right)$ with $\tilde{U}_H$ yields a $(\tilde{\alpha},\tilde{a}+2,0)$-block-encoding $\widetilde{\mathrm{REF}}$ of $\tilde{S}\in\mathbb{C}^{2^n \times 2^n}$ such that
    \begin{align}
        &\left\|\tilde{S}-\left(-S\left(\frac{H-\mu I}{\tilde{\alpha}+|\mu|};\frac{\epsilon}{2\tilde{\alpha}},\frac{\epsilon^\prime}{2}\right)\right)\right\| \nonumber \\
        \le & 4 d_{\tilde{\alpha},\epsilon,\epsilon^\prime} \sqrt{\left\|\frac{H-\mu I}{\tilde{\alpha}+|\mu|}-\frac{\tilde{H}-\mu I}{\tilde{\alpha}+|\mu|}\right\|} \nonumber \\
        \le &  \frac{4d_{\tilde{\alpha},\epsilon,\epsilon^\prime}}{\sqrt{\tilde{\alpha}}}\sqrt{\left\|H-\tilde{H}\right\|} \nonumber \\
        \le & \frac{\epsilon^\prime}{2}.
    \end{align}
    Consequently, replacing $U_H$ in $\mathrm{PROJ}\left(\mu,\frac{\epsilon}{2\tilde{\alpha}},\frac{\epsilon^\prime}{2}\right)$ with $\tilde{U}_H$ yields $\widetilde{\mathrm{PROJ}}\left(\mu,\frac{\epsilon}{2\tilde{\alpha}},\frac{\epsilon^\prime}{2}\right)$ satisfying
    \begin{widetext}
    \begin{align}
        &\left\|(\bra{0}^{\otimes(\tilde{a}+3)}\otimes I_{2^n})\widetilde{\mathrm{PROJ}}\left(\mu,\frac{\epsilon}{2\alpha},\epsilon^\prime\right)(\ket{0}^{\otimes(\tilde{a}+3)}\otimes I_{2^n}) -  (\bra{0}^{\otimes(\tilde{a}+3)}\otimes I_{2^n})\mathrm{PROJ}\left(\mu,\frac{\epsilon}{2\alpha},\frac{\epsilon^\prime}{2}\right)(\ket{0}^{\otimes(\tilde{a}+3)}\otimes I_{2^n})\right\| \nonumber \\
        = & \left\|(\bra{+}\bra{0}^{\tilde{a}+2}\otimes I_{2^n})\left(\ket{1}\bra{1}\otimes\left(\widetilde{\mathrm{REF}}-\mathrm{REF}\left(\mu,\frac{\epsilon}{2\tilde{\alpha}},\frac{\epsilon^\prime}{2}\right)\right)\right)(\ket{+}\ket{0}^{\otimes(\tilde{a}+2)}\otimes I_{2^n})\right\| \nonumber \\
        =&\frac{1}{2} \left\|(\bra{0}^{\tilde{a}+2}\otimes I_{2^n})\left(\widetilde{\mathrm{REF}}-\mathrm{REF}\left(\mu,\frac{\epsilon}{2\tilde{\alpha}},\frac{\epsilon^\prime}{2}\right)\right)(\ket{0}^{\otimes(\tilde{a}+2)}\otimes I_{2^n})\right\| \nonumber \\
        =&\frac{1}{2}\left\|\tilde{S}-\left(-S\left(\frac{H-\mu I}{\tilde{\alpha}+|\mu|};\frac{\epsilon}{2\tilde{\alpha}},\frac{\epsilon^\prime}{2}\right)\right)\right\| \le  \frac{\epsilon^\prime}{4},
    \end{align}
    \end{widetext}
    namely, Eq. \eqref{eq:PROJtil}.
    Here, at the first equality, we have used Eq. \eqref{eq:PROJDef} and the similar relationship between $\widetilde{\mathrm{PROJ}}\left(\mu,\frac{\epsilon}{2\alpha},\frac{\epsilon^\prime}{2}\right)$ and $\widetilde{\mathrm{REF}}$.

    Lastly, let us evaluate the number of uses of $O^H_\mathrm{row}$, $O^H_\mathrm{col}$, $O^H_\mathrm{ent}$ and elementary gates and the number of qubits in the quantum algorithm.
    As said above, we replace $\mathrm{PROJ}\left(\mu,\frac{\epsilon}{2\tilde{\alpha}},\epsilon^\prime\right)$ in the original algorithm in \cite{Lin2020nearoptimalground} with $\widetilde{\mathrm{PROJ}}\left(\mu,\frac{\epsilon}{2\tilde{\alpha}},\epsilon^\prime\right)$, which is yielded by replacing $U_H$ in $\mathrm{PROJ}\left(\mu,\frac{\epsilon}{2\tilde{\alpha}},\frac{\epsilon^\prime}{2}\right)$ with $\tilde{U}_H$.
    If we use $U_H$, the number of calls to it, the number of uses of other 1- and 2-qubit gates, and the number of qubits in $\mathrm{PROJ}\left(\mu,\frac{\epsilon}{2\tilde{\alpha}},\epsilon^\prime\right)$ and $\mathrm{PROJ}\left(\mu,\frac{\epsilon}{2\tilde{\alpha}},\frac{\epsilon^\prime}{2}\right)$ are of the same order.
    Since we use $O^H_\mathrm{row}$, $O^H_\mathrm{col}$ and $O^H_\mathrm{ent}$ in $\tilde{U}_H$ only $O(1)$ times, we get the evaluation \eqref{eq:GEEUHQuerySpOra} on the number of queries to these by simply replacing $\alpha$ in Eq. \eqref{eq:GEEUHQuery} with $\tilde{\alpha}$.
    In $\tilde{U}_H$, we use $O\left(n+\log^{5/2}\left(\frac{s^2\|H\|_\mathrm{max}}{\tilde{\epsilon}}\right)\right)$ additional 1- and 2-qubit gates, and multiplying this evaluation by the number of queries to $\tilde{U}_H$ and adding Eq. \eqref{eq:GEEGateNum} yields the total gate number evaluation in Eq. \eqref{eq:GEEGateNumSpOra}.
    $O\left(\log^{5/2}\left(\frac{s^2\|H\|_\mathrm{max}}{\tilde{\epsilon}}\right)\right)$ ancilla qubits are used in $\tilde{U}_H$, and summing up this and Eq. \eqref{eq:GEEQubitNum} with $a=\tilde{a}$ yields the total qubit number evaluation in Eq. \eqref{eq:GEEQubitNumSpOra}. 
    Since $U_{\ket{\phi_0}}$ is outside $\widetilde{\mathrm{PROJ}}\left(\mu,\frac{\epsilon}{2\tilde{\alpha}},\epsilon^\prime\right)$, we get the evaluation \eqref{eq:GEEUiniQuerySpOra} on the number of queries to this by simply replacing $\alpha$ in Eq. \eqref{eq:GEEUiniQuery} with $\tilde{\alpha}$.

\end{proof}

\subsection{Proof of Theorem \ref{th:main} \label{sec:ProofMain}}

Before presenting the rest of the proof, we give some lemmas.

\begin{lemma}
    Let $\mathbf{u}=(u_0,%...
    \cdots,u_{N_\mathrm{gr}-1}),\mathbf{v}=(v_0,%...
    \cdots,v_{N_\mathrm{gr}-1})\in\mathbb{R}^{N_\mathrm{gr}}$.
    Suppose that $\|\mathbf{v}_{n_\mathrm{gr}}\|_{n_\mathrm{gr}}=1$ and that $\|\mathbf{u}-\mathbf{v}\|_\mathrm{max}\le\epsilon$ with some $\epsilon\in\mathbb{R}$.
    Then,
    \begin{equation}
        \braket{\mathbf{u}|\mathbf{v}} \ge 1-2(U-L)^{d/2}\epsilon.
    \end{equation}
    \label{lem:uvInProBnd}
\end{lemma}

\begin{proof}
    $\|\mathbf{u}-\mathbf{v}\|_\mathrm{max}\le\epsilon$ implies that
    \begin{equation}
        u_i - \epsilon \le v_i \le u_i+\epsilon
        \label{eq:viBnd}
    \end{equation}
    for each $i\in[N_\mathrm{gr}]_0$, and that
    \begin{equation}
        \|\mathbf{u}\|\ge\|\mathbf{v}\| - \|\mathbf{u}-\mathbf{v}\|\ge \|\mathbf{v}\| - \sqrt{N_\mathrm{gr}}\epsilon.
        \label{eq:uvTriIneq}
    \end{equation}
    Besides, the Cauchy-Schwarz inequality implies that
    \begin{equation}
        \sum_{i=0}^{N_\mathrm{gr}-1}|u_i| \le \sqrt{N_\mathrm{gr}} \|\mathbf{u}\|.
        \label{eq:1norm2norm}
    \end{equation}
    Furthermore, by the definition of $\|\cdot\|_{n_\mathrm{gr}}$,
    \begin{equation}
        h^{d/2}\|\mathbf{v}\|=\|\mathbf{v}\|_{n_\mathrm{gr}}=1
        \label{eq:2normngrnorm}
    \end{equation}
    holds.
    Combining these, we have
    \begin{align}
        \braket{\mathbf{u}|\mathbf{v}} & = \frac{\sum_{i=0}^{N_\mathrm{gr}-1}u_iv_i}{\|\mathbf{u}\|\|\mathbf{v}\|} \nonumber \\
        & \ge \frac{\sum_{i=0}^{N_\mathrm{gr}-1}u_i(u_i-\mathrm{sgn}(u_i)\times\epsilon)}{\|\mathbf{u}\|\|\mathbf{v}\|} \nonumber \\
        & = \frac{\|\mathbf{u}\|}{\|\mathbf{v}\|} -\epsilon \frac{\sum_{i=0}^{N_\mathrm{gr}-1}|u_i|}{\|\mathbf{u}\|\|\mathbf{v}\|} \nonumber \\
        & \ge 1 - \frac{2\sqrt{N_\mathrm{gr}}\epsilon}{\|\mathbf{v}\|} \nonumber \\
        & \ge 1-2(U-L)^{d/2}\epsilon,
    \end{align}
    where we use Eq. \eqref{eq:viBnd} at the first inequality, Eqs. \eqref{eq:uvTriIneq} and \eqref{eq:1norm2norm} at the second inequality, and Eq. \eqref{eq:2normngrnorm} at the last equality.
\end{proof}

\begin{lemma}
    Let $\ket{\phi},\ket{\psi},\ket{\zeta}$ be quantum states on an $n_\mathrm{gr}$-qubit register.
    Suppose that $|\braket{\phi | \psi}|\ge\gamma$ for $\gamma\in(0,1)$.
    Then, $|\braket{\phi | \zeta}|\ge\frac{\gamma}{2}$ holds if
    \begin{equation}
        |\braket{\psi | \zeta}|\ge \eta(\gamma).
        \label{eq:PsiZeta}
    \end{equation}
    \label{lem:TriRelQState}
\end{lemma}

\begin{proof}
    We can write
    \begin{equation}
        \ket{\phi}=\alpha \ket{\psi} + \beta \ket{\psi_\perp}
    \end{equation}
    with $\alpha,\beta\in\mathbb{C}$ such that $|\alpha|\ge\gamma$ and $|\beta|\le\sqrt{1-\gamma^2}$, and a quantum state $\ket{\psi_\perp}$ orthogonal to $\ket{\psi}$.
    Besides, Eq. \eqref{eq:PsiZeta} implies that
    \begin{equation}
        |\braket{\zeta | \psi_\perp}| \le \sqrt{1-\eta^2(\gamma)}.
        \label{eq:zetapsiperp}
    \end{equation}
    Then, we have
    \begin{align}
        |\braket{\phi | \zeta}| & \ge |\alpha|\times|\braket{\zeta | \psi}|-|\beta|\times|\braket{\zeta | \psi_\perp}| \nonumber \\
        & \ge \gamma \eta(\gamma)-\sqrt{1-\gamma^2}\sqrt{1-\eta^2(\gamma)}.
    \end{align}
    By some algebra, we see that the last line is greater than or equal to $\frac{\gamma}{2}$. 
\end{proof}

Then, the rest of the proof of Theorem \ref{th:main} is as follows.

\begin{proof}[Rest of the proof of Theorem \ref{th:main}]

\ \\

\noindent \underline{\textbf{How to construct $O^{L_{n_\mathrm{gr}}}_\mathrm{row}$, $O^{L_{n_\mathrm{gr}}}_\mathrm{col}$ and $O^{L_{n_\mathrm{gr}}}_\mathrm{ent}$}}\\

    Let us start from $O^{L_{n_\mathrm{gr}}}_\mathrm{row}$.
    For $L_{n_\mathrm{gr}}$, whose sparsity is $2d+1$, $r_{ik}$ is given as\footnote{Strictly speaking, this expression for $r_{ik}$ holds for only $i$ such that none of the entries of $\sigma^{-1}(i)$ is $0$ or $n_\mathrm{gr}-1$, and otherwise the expression is slightly modified. However, such a handling is straightforward and thus we do not show the complete expression here for conciseness.}
    \begin{equation}
        r_{ik}=J(J^{-1}(i)+\mathbf{d}_{i,k}).
    \end{equation}
    Here, $\mathbf{d}_{i,k}\in\mathbb{R}^{d}$ is defined as
    \begin{equation}
        \mathbf{d}_{i,k}=
        \begin{cases}
            -\mathbf{e}_k & \mathrm{if} \ k=1,\cdots,d \\
            \mathbf{0}  & \mathrm{if} \ k=d+1 \\
            \mathbf{e}_{k-d-1} & \mathrm{if} \ k=d+2,\cdots,2d+1
        \end{cases}
        ,
    \end{equation}
    where $\mathbf{0}$ is the $d$-dimensional zero vector.
    $J:[n_\mathrm{gr}]_0^d\rightarrow[N_\mathrm{gr}]_0$ is the map in Eq. \eqref{eq:IdConv}, which is implemented by some additions and multiplications.
    $J^{-1}$ is its inverse, which is implemented by a sequence of divisions shown in Algorithm \ref{alg:Jinv}.
    Therefore, we see that we can implement $O^{L_{n_\mathrm{gr}}}_\mathrm{row}$ using $O(d)$ arithmetic circuits.
    $O^{L_{n_\mathrm{gr}}}_\mathrm{col}$ is implemented similarly.

    \begin{algorithm}[H]
    \begin{algorithmic}[1]
    \REQUIRE $K\in[N_\mathrm{gr}]_0$
    \ENSURE $k_1,\cdots,k_d$ such that $K=\sum_{i=1}^d n_\mathrm{gr}^{i-1} k_i$

    \STATE Set $K_{d}=K$.

    \FOR{$i=d,\cdots,2$}
    \STATE Divide $K_i$ by $n_\mathrm{gr}^{i-1}$, and let the quotient and remainder be $k_i$ and $K_{i-1}$, respectively. 

    \ENDFOR

    \STATE Divide $K_1$ by $n_\mathrm{gr}$, and let the quotient and remainder be $k_2$ and $k_1$, respectively. 
    
    \caption{$J^{-1}$}
    \label{alg:Jinv}
    \end{algorithmic}
    \end{algorithm}

    Next, let us consider the implementation of $O^{L_{n_\mathrm{gr}}}_\mathrm{ent}$.
    For $K,K^\prime\in[N_\mathrm{gr}]_0$, we can perform the following operation:
    \begin{widetext}
    \begin{align}
        &\ket{K}\ket{K^\prime}\ket{0}^{\otimes(4d+5)} \nonumber \\
        \rightarrow & \ket{K}\ket{K^\prime}\ket{\mathbf{k}}_\mathrm{bin}\ket{\mathbf{k}^\prime}_\mathrm{bin}\ket{0}^{\otimes(4d+3)} \nonumber \\
        \rightarrow & \ket{K}\ket{K^\prime}\ket{\mathbf{k}}_\mathrm{bin}\ket{\mathbf{k}^\prime}_\mathrm{bin}\Ket{a_0\left(\mathbf{x}^\mathrm{gr}_{\mathbf{k}}\right)}\Ket{a_1\left(\mathbf{x}^\mathrm{gr}_{\mathbf{k}}+\frac{h}{2}\mathbf{e}_1\right)}\Ket{a_1\left(\mathbf{x}^\mathrm{gr}_{\mathbf{k}}-\frac{h}{2}\mathbf{e}_1\right)} \cdots \Ket{a_d\left(\mathbf{x}^\mathrm{gr}_{\mathbf{k}}+\frac{h}{2}\mathbf{e}_d\right)}\Ket{a_d\left(\mathbf{x}^\mathrm{gr}_{\mathbf{k}}-\frac{h}{2}\mathbf{e}_d\right)} \nonumber \\
        & \ \otimes \ket{0}^{\otimes(2d+2)} \nonumber \\
        \rightarrow & \ket{K}\ket{K^\prime}\ket{\mathbf{k}}_\mathrm{bin}\ket{\mathbf{k}^\prime}_\mathrm{bin} \Ket{a_0\left(\mathbf{x}^\mathrm{gr}_{\mathbf{k}}\right)}\Ket{a_1\left(\mathbf{x}^\mathrm{gr}_{\mathbf{k}}+\frac{h}{2}\mathbf{e}_1\right)}\Ket{a_1\left(\mathbf{x}^\mathrm{gr}_{\mathbf{k}}-\frac{h}{2}\mathbf{e}_1\right)}\cdots\Ket{a_d\left(\mathbf{x}^\mathrm{gr}_{\mathbf{k}}+\frac{h}{2}\mathbf{e}_d\right)}\Ket{a_d\left(\mathbf{x}^\mathrm{gr}_{\mathbf{k}}-\frac{h}{2}\mathbf{e}_d\right)} \nonumber \\
        & \ \otimes \ket{1_{\mathbf{k}^\prime-\mathbf{k}=\mathbf{0}}}  \ket{1_{\mathbf{k}^\prime-\mathbf{k}=\mathbf{e}_1}}\ket{1_{\mathbf{k}^\prime-\mathbf{k}=-\mathbf{e}_1}}\cdots \ket{1_{\mathbf{k}^\prime-\mathbf{k}=\mathbf{e}_d}}\ket{1_{\mathbf{k}^\prime-\mathbf{k}=-\mathbf{e}_d}}\ket{0} \nonumber \\
        = & \ket{K}\ket{K^\prime}\ket{\mathbf{k}}_\mathrm{bin}\ket{\mathbf{k}^\prime}_\mathrm{bin}\Ket{a_0\left(\mathbf{x}^\mathrm{gr}_{\mathbf{k}}\right)}\Ket{a_1\left(\mathbf{x}^\mathrm{gr}_{\mathbf{k}}+\frac{h}{2}\mathbf{e}_1\right)}\Ket{a_1\left(\mathbf{x}^\mathrm{gr}_{\mathbf{k}}-\frac{h}{2}\mathbf{e}_1\right)} \cdots\Ket{a_d\left(\mathbf{x}^\mathrm{gr}_{\mathbf{k}}+\frac{h}{2}\mathbf{e}_d\right)}\Ket{a_d\left(\mathbf{x}^\mathrm{gr}_{\mathbf{k}}-\frac{h}{2}\mathbf{e}_d\right)} \nonumber \\
        & \ \otimes \left(1_{\mathbf{k}^\prime-\mathbf{k}\ne\mathbf{0},\pm\mathbf{e}_1,\cdots,\pm\mathbf{e}_d}\ket{0}^{\otimes (2d+1)}\ket{0}+1_{\mathbf{k}^\prime-\mathbf{k}=\mathbf{0}}\ket{1}\ket{0}^{\otimes 2d}\ket{0}\right.  \nonumber \\
        & \quad +1_{\mathbf{k}^\prime-\mathbf{k}=\mathbf{e}_1}\ket{0}\ket{1}\ket{0}\ket{0}^{\otimes (2d-2)}\ket{0} +1_{\mathbf{k}^\prime-\mathbf{k}=-\mathbf{e}_1}\ket{0}\ket{0}\ket{1}\ket{0}^{\otimes (2d-2)}\ket{0}+\cdots   \nonumber \\
        & \quad \left. +1_{\mathbf{k}^\prime-\mathbf{k}=\mathbf{e}_d}\ket{0}\ket{0}^{\otimes (2d-2)}\ket{1}\ket{0}\ket{0}+1_{\mathbf{k}^\prime-\mathbf{k}=-\mathbf{e}_d}\ket{0}\ket{0}^{\otimes (2d-2)}\ket{0}\ket{1}\ket{0}\right) \nonumber \\
        \rightarrow & \ket{K}\ket{K^\prime}\ket{\mathbf{k}}_\mathrm{bin}\ket{\mathbf{k}^\prime}_\mathrm{bin}\Ket{a_0\left(\mathbf{x}^\mathrm{gr}_{\mathbf{k}}\right)}\Ket{a_1\left(\mathbf{x}^\mathrm{gr}_{\mathbf{k}}+\frac{h}{2}\mathbf{e}_1\right)}\Ket{a_1\left(\mathbf{x}^\mathrm{gr}_{\mathbf{k}}-\frac{h}{2}\mathbf{e}_1\right)}\cdots\Ket{a_d\left(\mathbf{x}^\mathrm{gr}_{\mathbf{k}}+\frac{h}{2}\mathbf{e}_d\right)}\Ket{a_d\left(\mathbf{x}^\mathrm{gr}_{\mathbf{k}}-\frac{h}{2}\mathbf{e}_d\right)} \nonumber \\
        & \ \otimes \left(1_{\mathbf{k}^\prime-\mathbf{k}\ne\mathbf{0},\pm\mathbf{e}_1,\cdots,\pm\mathbf{e}_d}\ket{0}^{\otimes (2d+1)}\ket{0} +1_{\mathbf{k}^\prime-\mathbf{k}=\mathbf{0}}\ket{1}\ket{0}^{\otimes 2d}\Ket{(L_{n_\mathrm{gr}})_{K,K^\prime}}  \right.\nonumber \\
        & \quad +1_{\mathbf{k}^\prime-\mathbf{k}=\mathbf{e}_1}\ket{0}\ket{1}\ket{0}\ket{0}^{\otimes (2d-2)}\Ket{(L_{n_\mathrm{gr}})_{K,K^\prime}}+1_{\mathbf{k}^\prime-\mathbf{k}=-\mathbf{e}_1}\ket{0}\ket{0}\ket{1}\ket{0}^{\otimes (2d-2)}\Ket{(L_{n_\mathrm{gr}})_{K,K^\prime}}+\cdots \nonumber \\
        & \quad \left.+1_{\mathbf{k}^\prime-\mathbf{k}=\mathbf{e}_d}\ket{0}\ket{0}^{\otimes (2d-2)}\ket{1}\ket{0}\Ket{(L_{n_\mathrm{gr}})_{K,K^\prime}} +1_{\mathbf{k}^\prime-\mathbf{k}=-\mathbf{e}_d}\ket{0}\ket{0}^{\otimes (2d-2)}\ket{0}\ket{1}\Ket{(L_{n_\mathrm{gr}})_{K,K^\prime}} \right) \nonumber \\
        = & \ket{K}\ket{K^\prime}\ket{\mathbf{k}}_\mathrm{bin}\ket{\mathbf{k}^\prime}_\mathrm{bin} \Ket{a_0\left(\mathbf{x}^\mathrm{gr}_{\mathbf{k}}\right)}\Ket{a_1\left(\mathbf{x}^\mathrm{gr}_{\mathbf{k}}+\frac{h}{2}\mathbf{e}_1\right)}\Ket{a_1\left(\mathbf{x}^\mathrm{gr}_{\mathbf{k}}-\frac{h}{2}\mathbf{e}_1\right)} \cdots\Ket{a_d\left(\mathbf{x}^\mathrm{gr}_{\mathbf{k}}+\frac{h}{2}\mathbf{e}_d\right)}\Ket{a_d\left(\mathbf{x}^\mathrm{gr}_{\mathbf{k}}-\frac{h}{2}\mathbf{e}_d\right)} \nonumber \\
        & \ \otimes \ket{1_{\mathbf{k}^\prime-\mathbf{k}=\mathbf{0}}}\ket{1_{\mathbf{k}^\prime-\mathbf{k}=\mathbf{e}_1}}\ket{1_{\mathbf{k}^\prime-\mathbf{k}=-\mathbf{e}_1}}\cdots\ket{1_{\mathbf{k}^\prime-\mathbf{k}=\mathbf{e}_d}}\ket{1_{\mathbf{k}^\prime-\mathbf{k}=-\mathbf{e}_d}}\ket{(L_{n_\mathrm{gr}})_{K,K^\prime}} \nonumber \\
        \rightarrow & \ket{K}\ket{K^\prime}\ket{0}^{\otimes(4d+5)}\ket{(L_{n_\mathrm{gr}})_{K,K^\prime}}.
        \label{eq:OLent}
    \end{align}
    \end{widetext}
    Here, the transformation at the first arrow is done by the circuit for $J^{-1}$, where $\mathbf{k}=J^{-1}(K)$ and $\mathbf{k}^\prime=J^{-1}(K^\prime)$.
    At the second arrow, we make $O(d)$ uses of $O_{a_0},\cdots,O_{a_d}$.
    At the third arrow, we use $O(d)$ arithmetic circuits, and at the fourth arrow, we use controlled versions of arithmetic circuits to compute the entries of $L_{n_\mathrm{gr}}$ with the values of $a_0,%...
    \cdots,a_d$ computed in the previous step according to Eq. \eqref{eq:FDMat}.
    The last arrow in Eq. \eqref{eq:OLent} is uncomputation of the first and second ones.
    Note that the circuit for the operation in Eq. \eqref{eq:OLent} is nothing but $O^{L_{n_\mathrm{gr}}}_\mathrm{ent}$.

    Lastly, let us count the number of the queries to $O_{a_0},\cdots,O_{a_d}$ and arithmetic circuits in $O^{L_{n_\mathrm{gr}}}_\mathrm{row},O^{L_{n_\mathrm{gr}}}_\mathrm{col}$ and $O^{L_{n_\mathrm{gr}}}_\mathrm{ent}$.
    In $O^{L_{n_\mathrm{gr}}}_\mathrm{row}$ and $O^{L_{n_\mathrm{gr}}}_\mathrm{col}$, we use $O(d)$ arithmetic circuits.
    In $O^{L_{n_\mathrm{gr}}}_\mathrm{ent}$, we make $O(d)$ uses of $O_{a_0},\cdots,O_{a_d}$ and (controlled) arithmetic circuits.\\

\noindent \underline{\textbf{Accuracy}}\\

Because of Theorem \ref{th:EigConv}, for $n_\mathrm{gr}$ in \eqref{eq:ngrForEps}, $\lambda^1_{n_\mathrm{gr}}$ is $\frac{\epsilon}{2}$-close to $\lambda_1$.
Besides, the output of $\proc{EstEig}\left(L_{n_\mathrm{gr}},\frac{\epsilon}{2},\delta\right)$ is a $\frac{\epsilon}{2}$-approximation of $\lambda^1_{n_\mathrm{gr}}$.
Therefore, the output of Algorithm \ref{alg:EigEstim} is an $\epsilon$-approximation of $\lambda_1$.\\

\noindent \underline{\textbf{Query complexity}}\\

The sparsity of $L_{n_\mathrm{gr}}$ is
\begin{equation}
    s=O(d).
    \label{eq:spL}
\end{equation}
Its max norm is bounded as
\begin{align}
    &\|L_{n_\mathrm{gr}}\|_\mathrm{max} \nonumber \\
    \le & \frac{2d a_\mathrm{max}}{h^2} + a_{0,\mathrm{max}} \nonumber \\
    =&O\left(\frac{d a_\mathrm{max}}{(U-L)^2} \times \max\left\{\frac{C^1_\mathcal{L}}{\epsilon},\frac{D^1_\mathcal{L}(U-L)^{\frac{d}{2}}}{1-\eta(\gamma)}\right\}+ a_{0,\mathrm{max}}\right)
    \label{eq:normL}
\end{align}
for $n_\mathrm{gr}$ in Eq. \eqref{eq:ngrForEps}.

Let us evaluate the overlap between $\ket{\mathbf{v}_{\tilde{f}_1,n_\mathrm{gr}}}$ and $\ket{\mathbf{v}^1_{n_\mathrm{gr}}}$.
Because of Theorem \ref{th:EigConv},
\begin{equation}
\|\mathbf{v}^1_{n_\mathrm{gr}}-\mathbf{v}_{f_1,n_\mathrm{gr}}\|_\mathrm{max}\le \frac{1-\eta(\gamma)}{2(U-L)^{d/2}}
\end{equation}
holds for $n_\mathrm{gr}$ in Eq. \eqref{eq:ngrForEps}.
Because of Lemma \ref{lem:uvInProBnd}, this leads to
\begin{equation}
    \left|\braket{\mathbf{v}^1_{n_\mathrm{gr}}|\mathbf{v}_{f_1,n_\mathrm{gr}}}\right|\ge\eta(\gamma).
\end{equation}
Then, because of Lemma \ref{lem:TriRelQState}, combining this and Eq. \eqref{eq:f1f1til} leads to 
\begin{equation}
    \left|\braket{\mathbf{v}^1_{n_\mathrm{gr}}|\mathbf{v}_{\tilde{f}_1,n_\mathrm{gr}}}\right|\ge\frac{\gamma}{2}.
\end{equation}

Let us incorporate the above observations into Corollary \ref{co:GEESparse}.
By using Eqs. \eqref{eq:spL} and \eqref{eq:normL} in Eqs. \eqref{eq:GEEUHQuerySpOra} and \eqref{eq:GEEUiniQuerySpOra} and replacing $\gamma$ with $\frac{\gamma}{2}$, we obtain estimations of the numbers of queries to $O^{L_{n_\mathrm{gr}}}_\mathrm{row},O^{L_{n_\mathrm{gr}}}_\mathrm{col},O^{L_{n_\mathrm{gr}}}_\mathrm{ent}$ and $O_{\tilde{f}_1,n_\mathrm{gr}}$ in $\proc{EstEig}\left(L_{n_\mathrm{gr}},\frac{\epsilon}{2},\delta\right)$ in Algorithm \ref{alg:EigEstim}, and, noting that
$O^{L_{n_\mathrm{gr}}}_\mathrm{row}$, $O^{L_{n_\mathrm{gr}}}_\mathrm{col}$ and $O^{L_{n_\mathrm{gr}}}_\mathrm{ent}$ consist of $O(d)$ uses of $O_{a_0},\cdots,O_{a_d}$ and arithmetic circuits, we get the query complexity estimations in Eqs. \eqref{eq:EigEstCompOai} and \eqref{eq:EigEstCompOf1}.
    
\end{proof}

\bibliography{reference}

\end{document}